\documentclass[acmsmall]{acmart}
\usepackage[utf8]{inputenc}
\AtBeginDocument{%
  }

\setcopyright{acmlicensed}
\acmDOI{10.1145/3729293}
\acmYear{2025}
\acmJournal{PACMPL}
\acmVolume{9}
\acmNumber{PLDI}
\acmArticle{190}
\acmMonth{6}
\received{2024-11-15}
\received[accepted]{2025-03-06}

\usepackage{caption}
\usepackage{booktabs}
\usepackage{enumitem}
\usepackage{mathtools}
\usepackage{semantic}
\usepackage{geometry}
\usepackage{braket}
\usepackage{xspace}
\usepackage{diagbox}
\usepackage{array}
\usepackage{mathpartir}
\usepackage{graphicx}
\usepackage{stmaryrd}
\usepackage{bm}

\newcommand{\Vmodels}{%
\mathrel{\reflectbox{{$\models$}}}}
\newcommand{\eqmodels}{\Vmodels\!\:\!\models}

\newcommand{\qskip}{\mathbf{skip}}
\newcommand{\qif}{\mathbf{if}}
\newcommand{\qelse}{\mathbf{else}}
\newcommand{\qthen}{\mathbf{then}}
\newcommand{\qend}{\mathbf{end}}
\newcommand{\qfor}{\mathbf{for}}

\newcommand{\qdo}{\mathbf{do}}
\newcommand{\qwhile}{\mathbf{while}}
\newcommand{\qinit}[1]{#1 \coloneq \ket{0}}
\newcommand{\qassign}[2]{#2 \coloneq #1}
\newcommand{\qut}[2]{#2 \,*\!\!= #1}
\newcommand{\qmeas}[2]{\qassign{\textbf{meas}[#2]}{#1}}
\newcommand{\qqif}[3]{\qif\ #1 \ \qthen\ #2 \ \qelse\ #3 \ \qend}
\newcommand{\qqwhile}[2]{\qwhile\ #1 \ \qdo \ #2 \ \qend}
\newcommand{\qqfor}[2]{\qfor \ i \in #1 \ \qdo \  #2  \ \qend}
\newcommand{\qprog}{\mathit{Prog}}

\newcommand{\qstate}[1]{\vert#1\rangle}

\newcommand{\sem}[1]{\llbracket#1\rrbracket}
\renewcommand{\>}{\rangle}
\renewcommand{\<}{\langle}
\newcommand{\veriq}{\textsf{Veri-QEC}\xspace}

\DeclareMathOperator{\Span}{span}
\DeclareMathOperator{\Tr}{tr}
\DeclareMathOperator{\Supp}{supp}

\newcommand{\bs}{\bm{s}}
\newcommand{\be}{\bm{e}}
\newcommand{\cD}{\mathcal{D}}

\newcommand{\cH}{\mathcal{H}}
\newcommand{\cE}{\mathcal{E}}
\newcommand{\cI}{\mathcal{I}}
\newcommand{\cL}{\mathcal{L}}

\newcommand{\cG}{\mathcal{G}}

\newcommand{\cS}{\mathcal{S}}

\newcommand{\sq}{\frac{1}{\sqrt{2}}}
\newtheorem{observation}{Observation}[section]

\usepackage{tikz}
\usepackage{tikz-qtree}
\usetikzlibrary{backgrounds}

\usepackage{multirow}

\title{Efficient Formal Verification of Quantum Error Correcting Programs}

\author{Qifan Huang}
\orcid{0009-0005-6548-4303}
\email{huangqf@ios.ac.cn}
\affiliation{
    \department{Key Laboratory of System Software (Chinese Academy of Sciences) and State Key Laboratory of Computer Science}
    \institution{Institute of Software, Chinese Academy of Sciences}
    \city{Beijing}
    \country{China}
}
\affiliation{%
    \institution{University of Chinese Academy of Sciences}
    \city{Beijing}
    \country{China}
}

\author{Li Zhou}
\authornote{Corresponding author: Li Zhou, Mingsheng Ying}
\orcid{0000-0002-9868-8477}
\email{zhouli@ios.ac.cn}
\email{zhou31416@gmail.com}
\affiliation{
    \department{Key Laboratory of System Software (Chinese Academy of Sciences) and State Key Laboratory of Computer Science}
    \institution{Institute of Software, Chinese Academy of Sciences}
    \city{Beijing}
    \country{China}
}

\author{Wang Fang}
\orcid{0000-0001-7628-1185}
\email{fangw@ios.ac.cn}
\affiliation{%
  \institution{School of Informatics, University of Edinburgh}
  \city{Edinburgh}
  \country{United Kingdom}
}

\author{Mengyu Zhao}
\orcid{0009-0001-8436-3532}
\email{zhaomy@ios.ac.cn}
\affiliation{
    \department{Key Laboratory of System Software (Chinese Academy of Sciences) and State Key Laboratory of Computer Science}
    \institution{Institute of Software, Chinese Academy of Sciences}
    \city{Beijing}
    \country{China}
}
\affiliation{%
    \institution{University of Chinese Academy of Sciences}
    \city{Beijing}
    \country{China}
}

\author{Mingsheng Ying}
\authornotemark[1]
\orcid{0000-0003-4847-702X}
\email{mingsheng.ying@uts.edu.au}
\affiliation{%
    \institution{University of Technology Sydney}
    \city{Sydney}
    \country{Australia}
}

\begin{document}

\begin{abstract}
  Quantum error correction (QEC) is fundamental for suppressing noise in quantum hardware and enabling fault-tolerant quantum computation. In this paper, we propose an efficient verification framework for QEC programs. We define an assertion logic and a program logic specifically crafted for QEC programs and establish a sound proof system. We then develop an efficient method for handling verification conditions (VCs) of QEC programs: for Pauli errors, the VCs are reduced to classical assertions that can be solved by SMT solvers, and for non-Pauli errors, we provide a heuristic algorithm. We formalize the proposed program logic in Coq proof assistant, making it a verified QEC verifier. Additionally, we implement an automated QEC verifier, \veriq, for verifying various fault-tolerant scenarios. We demonstrate the efficiency and broad functionality of the framework by performing different verification tasks across various scenarios. Finally, we present a benchmark of 14 verified stabilizer codes.
\end{abstract}



\maketitle

\section{Introduction}
\label{introduction}

    Beyond the current noisy intermediate scale quantum (NISQ) era~\cite{Preskill2018NISQ}, 
    fault-tolerant quantum computation is an indispensable step towards scalable quantum computation. Quantum error correcting (QEC) codes serve as a foundation for suppressing noise and implementing fault-tolerant quantum computation in noisy quantum hardware. There have been more and more experiments illustrating the implementation of quantum error correcting codes in real quantum processors~\cite{ryan2021realization, zhao2022realization, Acharya2023suppress, Bluvstein2024logicala, Bravyi2024highthreshold}. These experiments show the great potential of QEC codes to reduce noise. Nevertheless, the increasingly complex QEC protocols make it crucial to verify the correctness of these protocols before deploying them. 

    There have been several verification techniques developed for QEC programs. Numerical simulation, especially \textit{stabilizer-based simulation}~\cite{Aaronson2004improved,Simon2006fast,gidney2021stim} is extensively used for testing QEC programs. 
    While stabilizer-based simulations can efficiently handle QEC circuits with only Clifford operations ~\cite{nielsen2010quantum} compared to general methods~\cite{xu2023herculean}, showing the effectiveness and correctness of QEC circuits still requires millions or even trillions of test cases, which is the main bottleneck~\cite{gidney2021stim}. Recently, \textit{symbolic execution}~\cite{Fang2024symbolic} has also been applied to verify QEC programs. It is an automated approach designed to handle a large number of test cases and is primarily intended for bug reporting. However, it has limited functionality, such as the inability to reason about non-Clifford gates or propagation errors, and it remains slow when verifying correct instances.
    
    Program logic is another appealing verification technique. It naturally handles a class of instances simultaneously by expressing and reasoning about rich specifications in a mathematical way~\cite{grout2011digital}. Two recent works pave the way for using Hoare-style program logic for reasoning about QEC programs. Both works leverage the concept of stabilizer, which is critical in current QEC codes to develop their programming models. 
    Sundaram et al.~\cite{sundaram2022hoare} established a lightweight Hoare-like logic for quantum programs that treat stabilizers as predicates. Wu et al.~\cite{wu2021qecv, wu2024towards} studied the syntax and semantics of QEC programs by employing stabilizers as first-class objects. They proposed a program logic designed for verifying QEC programs with fixed operations and errors.
    Yet, at this moment, these approaches do not achieve usability for verifying large-scale QEC codes with complicated structures, in particular for real scenarios of errors that appear in fault-tolerant quantum computation. 

    \paragraph{\textbf{Technical challenge}}
        There are still critical challenges to the efficient verification of large-scale QEC programs, as summarized below.
        \begin{itemize}
        [leftmargin=0.5cm]
            \item \textit{A suitable hybrid program logic supporting backward reasoning.}
            QEC codes are designed to correct possible errors, making error modeling crucial for verification. To this end, it is necessary to introduce classical variables to describe errors and measurement outcomes, as well as properties like the maximum number of correctable errors. Backward reasoning is then desired since it gives a simple but complete rule for classical assignment,
            while forward reasoning needs additional universal quantifiers to ensure completeness. As discussed in~\cite{unruh2019quantum} and illustrated in Example \ref{exam-failure-disjunction}, interpreting $\vee$ as classical disjunction suffers from the incompleteness problem even for QEC codes, making it necessary to choose quantum logic as base logic, where, $\vee$ is interpreted as the sum of subspaces.
            \item \textit{Proving verification conditions generated by program logic.}
            Traditionally, after annotating the program, the program logic will generate verification conditions (entailment of assertion formulas).
            A complete and rigorous approach is to use formal proofs; however, this requires significant human effort.
            Another approach is to use efficient solvers to achieve automatic proofs.
            Unfortunately, quantum logic lacks efficient tools similar to SMT solvers: systematically handling quantum logic has been a longstanding challenge. On the one hand, the continuity of subspaces makes brute-force search ineffective, while on the other hand, the lack of distributive laws makes finding a (canonical) normal form particularly difficult. It remains unknown if assertion formulas for QEC codes can be efficiently processed.
        \end{itemize}

    \vspace{-0.12cm}

    \paragraph{\textbf{Contributions}}
        \begin{figure}
            \centering
            \vspace{-0.9em}
            \scalebox{.78}{\usetikzlibrary{shadows, backgrounds, positioning}

\definecolor{primary}{RGB}{41,128,185}    
\definecolor{secondary}{RGB}{52,152,219}  
\definecolor{accent1}{RGB}{46,204,113}    
\definecolor{accent2}{RGB}{241,196,15}    
\definecolor{accent3}{RGB}{231,76,60}     
\definecolor{textcolor}{RGB}{44,62,80}   

\begin{tikzpicture}
    [
        edge from parent/.style={
            draw=primary!80,-stealth,thick,
            edge from parent path={(\tikzparentnode.south)--+(0,-8pt)-|(\tikzchildnode)}
        },
        level distance={1.5cm},
        sibling distance=.3cm,
        every tree node/.style={align=center,,draw,rectangle,thick,},
    ]

    \Tree
    [.\node[fill=primary!30, draw=primary, rounded corners=2pt,very thick](q0){Program Logic for QEC Codes};
        [.\node[fill=accent1!25, draw=accent1!80,rounded corners=2pt,very thick](q1){Verified QEC Verifier \\ (Coq-based)};
            [.\node[fill=accent1!20](q2){Theory \\ Formalization}; ]
            [.\node[fill=accent1!20](q6){Interactive Pen- \\ and-paper Proof}; 
                [.\node[fill=accent1!15]{Scalable codes}; ]
            ]
        ]
        [.\node[fill=accent2!25, draw=accent2!80, rounded corners=2pt,very thick]{Automated QEC Verifier \\ \veriq (SMT-based)};
            [.\node[fill=accent2!20](q7){General Verification \\ for all error configurations};
                [.\node[fill=accent2!15](q3){Small scale codes \\ ($\sim$120 qubits)}; ]
            ]
            [.\node[fill=accent2!20](q8){Partial Verification \\ for user-provided error patterns};
                [ .\node[fill=accent2!15](q4){Medium scale codes \\ ($\sim$360 qubits)}; ]
            ]
        ]
    ]

    \draw[->,dotted,thick] (q2.north)+(-6mm,0) |- node[pos=.7,align=center,font=\small] {Enhance \\ confidence} (q0);

    \foreach \i/\Text in {0/{Theory},1/{Tools},3/{Performance}}{
        \node[anchor=west,font=\bfseries] (p\i) at ([xshift=-2.5cm]{q2}|-{q\i}) {\Text};
    }
    \coordinate[yshift=5mm] (p1) at (p1.north west);
    \coordinate[yshift=5mm] (p3) at (p3.north west);
    \coordinate[xshift=2.6cm] (q4) at (q4);
    \foreach \i in {1,3}{
        \draw[dashed] (p\i) -- ({q4}|-{p\i});
    }
\end{tikzpicture}}
            \vspace{-0.9em}
            \caption{Overall structure of our verification framework for QEC programs.
              }
            \vspace{-0.5em}
            \label{fig:framework}
        \end{figure}
        We propose a formal verification framework, summarized in Fig. \ref{fig:framework}, by proposing theoretical solutions to the above challenges, together with two implementations, (i) \textit{the Coq-based verified QEC verifier} and (ii) \textit{the SMT-based automatic QEC verifier \veriq}, that ensure and illustrate the effectiveness of our theory. In detail, we contribute:
    
        \vspace{-0.2em}
    
        \begin{itemize}
        [leftmargin=0.5cm]
            \item \textit{Assertion logic and program logic}  (Section \ref{assert-logic} and \ref{prog-logic}). Following~\cite{sundaram2022hoare, wu2021qecv}, we use Pauli expressions as atomic propositions and interpret them as the $+1$-eigenspace of the corresponding Pauli operator. We additionally introduce classical variables and interpret logical connectives based on quantum logic, e.g., interpreting $\vee$ as the sum of subspaces rather than as a union. Adopting the semantics for classical-quantum from~\cite{feng2021quantum}, we establish a sound proof system for quantum programs.
        
            \item \textit{Efficient handling of verification condition of QEC code} (Section \ref{sec:VC-framework}). The verification condition generated by a QEC code is typically of the form
            \begin{equation}
             \left( P_1 \wedge \dots \wedge P_n \right) \wedge \Phi_c \models \bigvee\nolimits_{\bs\in \{0,1\}^n} \left ((-1)^{f_1(\bs)} P_1^\prime\wedge \dots \wedge (-1)^{f_n(\bs)} P_n^\prime\right), 
            \end{equation}
            where $P_i, P_i'$ are Pauli expressions and $\Phi_c$ is a classical assertion. Progressing from simple to complex, we deal with the following cases: 
            1). $\{P_i'\}\subseteq \{P_j\}$. Then it is equivalent to compare phase, which can be efficiently solved by an SMT solver.
            2). All $P_i$ and $P_j'$ commute. Then employ the fact that $P_i' = (-1)^{\alpha_i}\prod_{k\in K_i} P_k$ since $\{P_i\}$ is a minimal generating set and $P\wedge Q = P\wedge QP$~\cite{sundaram2022hoare} to reduce it to case 1).
            3). A non-commuting pair $P_i$ and $P_j'$ exists. Then a heuristic algorithm is proposed to recursively eliminate $P_j'$ from $\{P_i'\}$ based on the facts $(P\wedge Q)\vee (\neg P\wedge Q) = Q$ if $P$ commute with $Q$, and finally reduce it to case 2).
        
            \item \textit{A verified QEC verifier} (Section \ref{sec:tool-implementation}). We formalize our program logic in Coq proof assistant~\cite{coq} based on CoqQ~\cite{zhou2023coqq}, i.e., proving the soundness of the proof system. This enhances confidence in the designed program logic. As a byproduct, this also allows us to manually formalize pen-and-paper proofs of scalable codes.
            
            \item \textit{Automatic QEC verifier \veriq} (Section \ref{sec:tool-implementation} and \ref{sec:evaluation}).  \veriq is a practical tool developed in Python with the aid of \textsc{Z3} and \textsc{CVC5} SMT solvers~\cite{de2008z3, DBLP:conf/tacas/BarbosaBBKLMMMN22}. 
            \veriq supports verification in various scenarios, from standard errors to propagation errors or errors in correction steps, and from one cycle of QEC code to fault-tolerant implementation of small logical circuits. We examine \veriq on $14$ QEC codes selected from the stabilizer code family with $5$-$361$ qubits and perform different verification tasks based on the type of code and distance. Typical performance on surface codes includes: general verification for all error configurations up to $121$ qubits within $\sim 200$ minutes, and partial verification for user-provided error constraints up to $361$ qubits within $\sim 100$ minutes.
        \end{itemize}
    \vspace{-0.12cm}
    \paragraph{\textbf{Comparison to existing works}}
        Here we compare our work with works related to verifying QEC programs and leave the general discussion of related works in Section~\ref{sec:related_work}.
        Thanks to the efficiency of the stabilizer formalism in describing Clifford operations used in QEC programs, several works~\cite{Rand2021gottesman,wu2021qecv,wu2024towards,Rand2021StaticAO,sundaram2022hoare} utilize stabilizers as assertions in quantum programs.
        Among them, \citet{Rand2021gottesman,Rand2021StaticAO} built stabilizer formalism by designing a type system of Gottesman types, upon which \citet{sundaram2022hoare} further established a Hoare-like logic to characterize quantum programs consisting of Clifford gates, $T$ gate and measurements.
        The proof system was built in forward reasoning; thus the disjoint union `$\uplus$' is employed to describe the post-measurement state.
        \citet{wu2021qecv} focused more on QEC. They designed a programming language with a stabilizer constructor in the syntax, specifically for QEC programs. This programming language faithfully captures the implementation of QEC protocols. To verify the correctness of QEC programs more efficiently while ensuring the accurate characterization of their properties, they designed an assertion logic using sums of stabilizers as atomic propositions and \emph{classical} logical connectives. Given fixed operations, errors, and exact results of the decoder, this framework can effectively prove the correctness of a given QEC program.
        
        Compared with prior works, our verification framework stands out by incorporating classical variables into both programs and assertions.
        Our assertion language enables simultaneous reasoning about properties of subspaces and a family of quantum states, such as logical computational basis states, which previous QEC program logic could only handle individually. Together with the classical variables in the program, our framework can model and verify the conditions of errors that previous work cannot reason about, e.g. the maximum correctable number of errors. Our program logic provides strong flexibility and efficiency to insert errors anywhere in the QEC program, such as before and after logic operators and within correction steps, and then verify the correctness. This capability is crucial for the subsequent step of verifying the implementation of fault-tolerant quantum computing.

\section{Motivating example: The Steane code}
We introduce a motivating example, the Steane code, which is widely used in quantum computers~\cite{Nigg2014quantum,ryananderson2022implementingfaulttolerantentanglinggates,Bluvstein2022,Bluvstein2024logicala} to construct quantum circuits. A recent work~\cite{Bluvstein2024logicala} demonstrates the use of Steane code to implement fault-tolerant logical algorithms in reconfigurable neutral-atom arrays. 
We aim to demonstrate the basic concepts of our formal verification framework through the verification of Steane code.

\subsection{Basic Notations and Concepts}

    \paragraph{Quantum state.}
    Any state $|\psi\>$ of quantum bit (qubit) can be represented by a two-dimensional vector 
    $
    \left(\begin{smallmatrix}
      \alpha \\
      \beta
    \end{smallmatrix}\right)$ with $\alpha, \beta \in \mathbb{C}$ satisfying $|\alpha|^2 + |\beta|^2 = 1$. 
    Frequently used states include computational bases
    $|0\>\triangleq
    \left(\begin{smallmatrix}
      1 \\
      0
    \end{smallmatrix}\right)$ and 
    $|1\>\triangleq
    \left(\begin{smallmatrix}
      0 \\
      1
    \end{smallmatrix}\right)$, and $|\pm\rangle = \frac{1}{\sqrt{2}}(|0\rangle\pm|1\rangle)$.
    The computational basis of an $n$-qubit system is $|\bm{s}\>\triangleq|s_1s_2\cdots s_n\>$ where $\bm{s}$ is a bit string, and any state $|\psi\>$ is a superposition $|\psi\> = \sum_{\bm{s}\in\{0,1\}^n}a_{\bm{s}}|\bm{s}\>$.
    
\vspace{-0.01cm}

    \paragraph{Unitary operator.}
    The evolution of a (closed) quantum system is modeled as a unitary operator, aka quantum gate for qubit systems. Here we list some of the commonly used quantum gates:
    {\small
    \begin{align*}
    I = \begin{pmatrix}1 & 0 \\ 0 & 1\end{pmatrix} \quad 
    X = \begin{pmatrix}0 & 1 \\ 1 & 0\end{pmatrix} \quad 
    Y = \begin{pmatrix}0 & -i \\ i & 0\end{pmatrix} &\quad 
    Z = \begin{pmatrix}1 & 0 \\ 0 & -1 \end{pmatrix} \quad
    H = \frac{1}{\sqrt{2}}\begin{pmatrix}1 & 1 \\ 1 & -1\end{pmatrix} \quad  
    S = \begin{pmatrix} 1 & 0 \\ 0 & i \end{pmatrix} \\\
    T = \begin{pmatrix} 1 & 0 \\ 0 & e^{\frac{i\pi}{4}} \end{pmatrix} \quad
    CNOT = \left({\large\begin{smallmatrix}1 & 0 & 0 & 0 \\ 0 & 1 & 0 & 0 \\0 & 0 & 0 & 1 \\ 0 & 0 & 1 & 0 \end{smallmatrix}}\right) \quad &
    CZ = \left({\large\begin{smallmatrix}1 & 0 & 0 & 0 \\ 0 & 1 & 0 & 0 \\0 & 0 & 1 & 0 \\ 0 & 0 & 0 & -1 \end{smallmatrix}}\right) \quad 
    iSWAP = \left({\large\begin{smallmatrix}1 & 0 & 0 & 0 \\ 0 & 0 & -i & 0 \\0 & -i & 0 & 0 \\ 0 & 0 & 0 & 1 \end{smallmatrix}}\right).
    \end{align*}}
    The evolution is computed by matrix multiplication, for example, $H$ gate transforms $|0\>$ to $H|0\> = \frac{1}{\sqrt{2}}\left(\begin{smallmatrix}
      1 & 1 \\ 1 & -1
    \end{smallmatrix}\right)\left(\begin{smallmatrix}
      1 \\
      0
    \end{smallmatrix}\right) = \frac{1}{\sqrt{2}}\left(\begin{smallmatrix}
      1 \\
      1
    \end{smallmatrix}\right) = |+\>$.

    \paragraph{Projective measurement.}
    We here consider the boolean-valued projective measurement $M = \{P_0, P_1\}$ with projections $P_0$ and $P_1$ such that $P_0+P_1 = I$. Performing $M$ on a given state $|\psi\>$, with probability $p_m = |P_m|\psi\>|^2$ we get $m$ and post-measurement state $\frac{P_m|\psi\>}{\sqrt{p_m}}$ for $m = 0, 1$.

\vspace{-0.08cm}

    \paragraph{Pauli group and Clifford gate.} 
    The \textit{Pauli group} on $n$ qubits $\mathcal{P}_n$ consists of all Pauli strings $g$ which are represented by the tensor product of $n$ Pauli or identity matrices with multiplicative factor $\pm 1,\pm i$, i.e., $i^t p_1\otimes \cdots \otimes p_n$, where $p_i\in \{I, X, Y, Z\}, t\in\{0,1,2,3\}$. 
    A state $\qstate{\psi}$ is stabilized by $g\in \mathcal{P}_n$ (or a subset $S\subseteq \mathcal{P}_n$) , if $g\qstate{\psi} = \qstate{\psi}$ (or $\forall\, g\in S,\ g\qstate{\psi} = \qstate{\psi}$).
    The measurement outcome of the corresponding projective measurement $M_g$ is always $0$ iff $|\psi\>$ is a stabilizer state of $g$. A unitary $V$ is a \emph{Clifford gate}, if for any Pauli string $g$,  $VgV^{\dag}$ is still a Pauli string. All Clifford gates form the Clifford group, and can be generated by $H, S$, and $CNOT$.

\vspace{-0.04cm}

    \paragraph{Stabilizer code.} 
    An $[[n,k,d]]$ stabilizer code $\mathcal{C}$ is a subspace of the $n$-qubit state space, defined as the set (aka codespace) of states stabilized by an abelian subgroup $S$ (aka stabilizer group) of Pauli group $\mathcal{P}_n$, with a minimal representation in terms of $n-k$ independent and commuting generators $\langle g_1, \dots, g_{n-k}\rangle$ requiring $-I \notin S$.
    The codespace of $\mathcal{C}$ is of dimension $2^k$ and thus able to encode $k$ logical qubits into $n$ physical qubits.
    With additional $k$ logical operators $\bar{Z}_1, \cdots, \bar{Z}_k$ that are independent and commuting with each other and $S$, we can define a $k$-qubit logical state $\qstate{z_1,\dots,z_k}_L$ as the state stabilized by $\langle g_1, \dots, g_{n-k}, (-1)^{z_1}\bar{Z}_1, \dots, (-1)^{z_k}\bar{Z}_k\rangle$ with $z_i \in \{0,1\}$.
    We can further construct $\bar{X}_1, \dots, \bar{X}_k$ such that $\bar{X}_i$ commute with $g\in S$ and $\bar{X}_i\bar{Z}_j = (-1)^{\delta_{ij}} \bar{Z}_j\bar{X}_i$ for all $i,j\in \{1,\cdots,k\}$, and regard $\bar{Z}_i$ (or $\bar{X}_i$) as logical $Z$ (or $X$) gate acting on $i$-th logical qubit.
    $d$ is the code distance, i.e., the minimum (Hamming) weight of errors that can go undetected by the code.

\subsection{The [[7,1,3]] Steane code }\label{example-steane}
The Steane code encodes a logical qubit using 7 physical qubits. The code distance is 3, therefore it is the smallest CSS code~\cite{calderbank1996good} that can correct any single-qubit Pauli error. The generators $g_1,\ldots,g_6$, and logical operators $\bar{X}$ and $\bar{Z}$ of Steane code are as follows: 

\begin{align*}
 g_1 & \coloneq X_1X_3X_5X_7 & g_2 & \coloneq X_2X_3X_6X_7 &  g_3 & \coloneq X_4X_5X_6X_7 & \bar{X} & \coloneq X_1X_2X_3X_4X_5X_6X_7  \\
 g_4  & \coloneq Z_1Z_3Z_5Z_7 & g_5  & \coloneq Z_2Z_3Z_6Z_7 & g_6  & \coloneq Z_4Z_5Z_6Z_7 & \bar{Z}  & \coloneq Z_1Z_2Z_3Z_4Z_5Z_6Z_7.
\end{align*}

In Table \ref{tab:prog}, we describe the implementations of logical Clifford operations and error correction procedures using the programming syntax introduced in Section \ref{prog-logic}.

As a running example, we analyze a one-round error correction process in the presence of single-qubit Pauli $Y$ errors, as well as the Hadamard $H$ error and $T$ error serving as instances of non-Pauli errors. First, we inject propagation errors controlled by Boolean-valued indicators $\{e_{pi}\}$ at the beginning. A propagation error simulates the leftover error from the previous error correction process, which must be considered and analyzed to achieve large-scale fault-tolerant computing.
Next, a logical operation $H$ is applied followed by the standard error injection controlled by indicators $\{e_i\}$. Formally, $[e_i]\qut{U}{q_i}$ means applying the error $U$ on $q_i$ if $e_i=1$, and skipping otherwise. Afterwards, we measure the system according to generators of the stabilizer group, compute the decoding functions $f_{x,i}$ and $f_{z,i}$, and finally perform correction operations. The technical details of the program can be found in Section \ref{case-steane} and Appendix \ref{case-steane-app}. 
\begin{table}
  \caption{Program Implementations of logical operation and error correction using a 7-qubit Steane code.}
  \label{tab:prog}
  \renewcommand{\arraystretch}{0.95}
  \begin{tabular}{|c|l|l|l|}
  \hline
  \multicolumn{2}{|c}{Logical Operation} & \multicolumn{2}{|c|}{Error Correction} \\
    \hline
    & Command & Explanation & \textbf{Steane}$(E, H)$ \quad $E\in \{Y, H, T\}$  \\
    \hline
    $H$ & $\qfor \ i \in 1\dots 7 \ \qdo\ $ 
      & Propagation Error & $\qqfor{1\dots 7}{\qut{E}{[e_{pi}]q_i}}$ \\ 
      & $\quad \qut{H}{q_i} \ \qend$ 
      & Logical operation $H$ & $\qfor \ i \in 1\dots 7 \ \qdo\ \qut{H}{q_i}\ \qend$ \\
    $S$ & $\qfor \ i \in 1\dots 7 \ \qdo\ $ 
      & Error injection & $\qqfor{1\dots 7}{\qut{E}{[e_i]q_i}} $ \\
      & $\quad \qut{Z}{q_i}\fatsemi\qut{S}{q_i}$ 
      & Syndrome meas & $\qqfor{1\dots 6}{\qassign{\mathbf{meas}[g_i]}{s_i}}$ \\
      & $\qend$ 
      & Call decoder for $Z$ &$\qassign{f_{z}(s_1,s_2,s_3)}{z_1,\ldots,z_7}$ \\
    $CNOT$ & $\qfor \ i \in 1 \dots 7 \ \qdo$ 
      & Call decoder for $X$ & $\qassign{f_{x}(s_4,s_5,s_6)}{x_1,\ldots,x_7}$ \\
      & $\quad \qut{CNOT}{q_i, q_{i+7}}$
      & Correction for $X$ & $\qqfor{1\dots 7}{\qut{X}{[x_i]q_i}}$ \\
      & $\qend$ 
      & Correction for $Z$ & $\qqfor{1\dots 7}{\qut{Z}{[z_i]q_i}}$\\
  \hline
\end{tabular}

\end{table}

The correctness formula for the program $\textbf{Steane}(Y,H)$ can be stated as the Hoare triple\footnote{Following the adequacy theorem stated in~\cite{Fang2024symbolic}, the correctness of the program is guaranteed as long as it holds true for only two predicates $(-1)^b Z \wedge \bigwedge_i g_i$ and  $(-1)^b X \wedge \bigwedge_i g_i$. Furthermore, since Steane code is a self-dual CSS code, the logical X and Z operators share the same form. Therefore only logical Z is considered here.}: 

\begin{equation}
\left\{\bigg(\sum_{i=1}^{7} (e_i + e_{pi}) \leq 1\bigg) \bigwedge  \left((-1)^{b}\bar{X} \wedge g_1\wedge \cdots \wedge g_6\right)\right\}
\textbf{Steane}(Y,H) \left\{(-1)^{b}\bar{Z} \wedge g_1\wedge \cdots \wedge g_6\right\}.
\label{corr-steane}
\end{equation}

Here, $b$ is a parameter denoting the phase of the logical state, e.g., $b = 0$ for initial state $|+\>_L$ (i.e., state stabilized by $\bar{X} \wedge g_1\wedge \cdots \wedge g_6$) and final state $|0\>_L$ (i.e., state stabilized by $\bar{Z} \wedge g_1\wedge \cdots \wedge g_6$). The correctness formula claims that if there is at most one $U$ error $(\sum_{i=1}^{7} (e_i + e_{pi}) \leq 1)$, then the program transforms $|+\>_L$ to $|0\>_L$ (and $\vert  -\>_L$ to $|1\>_L$), exactly the same as the error-free program that execute logical Hadamard gate $H$.

It appears hard to verify Eqn. (\ref{corr-steane}) in previous works. 
\cite{wu2021qecv,wu2024towards} can only handle fixed Pauli errors while $\mathbf{Steane}$ involves non-Pauli errors $T$ with flexible positions. \cite{sundaram2022hoare, Rand2021gottesman} do not introduce classical variables and thus cannot represent flexible errors nor reason about the constraints or properties of errors. \citet{Fang2024symbolic} cannot handle non-Clifford gates, since non-Clifford gates change the stabilizer generators (Pauli operators) into linear combinations of Pauli operators, which are beyond their scope.

In the following sections, we will verify Eqn. \eqref{corr-steane} by first deriving a precondition $A^\prime$ (see Eqn. \eqref{wp-stabilizer} for $Y$ error and Eqn. \eqref{eqn:T-err-wp} for $T$ error) by applying the inference rules from Fig. \ref{infer-rule}, and then proving the verification condition $A \models A^\prime$ based on the techniques proposed in Section \ref{verify-cond}.

\section{An Assertion logic for QEC programs}\label{assert-logic}
In this section, we introduce a hybrid classical-quantum assertion logic on which our verification framework is based. 

\subsection{Expressions} 
For simplicity, we do not explicitly provide the syntax of expressions of Boolean (denoted by $BExp$); see Appendix \ref{comp-expr} for an example. Their value is fully determined by the state of the classical memory $m\in \mathtt{CMem}$, which is a map from variable names to their values. Given a state  $m$ of the classical memory, we write $\sem{\cdot}_m$ for the semantics of basic expressions in state $m$.

A special class of expressions was introduced by~\cite{wu2024towards, sundaram2022hoare}, namely Pauli expressions. In particular, for reasoning about QEC codes with $T$ gates, \citet{sundaram2022hoare} suggests extending basic Pauli groups with addition and scalar multiplication with factor from the ring $\mathbb{Z}[1/\sqrt{2}] \triangleq \{\,x+y/\sqrt{2}\mid x,y\in \mathbb{Z}\,\} = \{\,(x+y\sqrt{2})/2^t\mid t\in\mathbb{N}, x,y\in \mathbb{Z}\,\}$. 
We adopt a similar syntax of expressions in the ring $\mathbb{Z}[\frac{1}{\sqrt{2}}]$ and Pauli expressions for describing generators of stabilizer groups:
\begin{align}
    SExp: && S &\Coloneqq{} (-1)^b \mid \sqrt{2} \mid S /2^t \mid S_1 + S_2 \mid - S \mid S_1S_2
    && \mbox{syntax for ring $\mathbb{Z}[\frac{1}{\sqrt{2}}]$}.
    \label{s2ring-exp}\\
    PExp: &&  
    P &\Coloneqq{}  p_r \mid  sP \mid P_1P_2 \mid P_1 + P_2
    &&\mbox{syntax for Pauli group with $s\in SExp$.}
    \label{pauli-exp}
\end{align}

In $SExp$, $b$ is a Boolean expression and $t$ is an expression of  natural numbers.
In $PExp$, $p_r$ is an elementary gate defined as $p \in \{ X, Y, Z \}$ with $r$ being a constant natural number indicating the qubit that $p$ acts on. $SExp$ and $PExp$ are  interpreted inductively as follows:
\begin{align*}
    &\sem{(-1)^b}_m \triangleq (-1)^{\sem{b}_m}, \quad \sem{\sqrt{2}}_m \triangleq \sqrt{2}, \quad
    \sem{s/2^t}_m \triangleq \frac{\sem{s}_m}{2^{\sem{t}_m}},\\
    &\sem{s_1 + s_2}_m \triangleq \sem{s_1}_m + \sem{s_2}_m,\quad
    \sem{-s}_m \triangleq -\sem{s}_m,\quad
    \sem{s_1s_2}_m \triangleq \sem{s_1}_m\sem{s_2}_m \\
    &\sem{p_r}_m \triangleq I_1\otimes\cdots\otimes I_{r-1}\otimes p_r\otimes I_{r+1}\otimes\cdots \otimes I_n\\
    &\sem{sP}_m = \sem{s}_m\sem{P}_m, \quad
    \sem{P_1P_2}_m \triangleq \sem{P_1}_m\sem{P_2}_m, \quad \sem{P_1 + P_2}_m \triangleq \sem{P_1}_m + \sem{P_2}_m.
\end{align*}
Here, $p_r$ is interpreted as a global gate by lifting it to the whole system, with $\otimes$ being the tensor product of linear operators, i.e., the Kronecker product if operators are written in matrix form. Such lifting is also known as cylindrical extension, and we sometimes omit explicitly writing out it.
Note that it is redundant to introduce the syntax of the tensor product $p_{r_1}\otimes p_{r_2}$ with different $r_1, r_2$, since
$\sem{p_{r_1} \otimes p_{r_2}}_m = I_1\otimes\cdots\otimes I_{r_1-1}\otimes p_{r_1}\otimes I_{r_1+1}\otimes\cdots\otimes I_{r_2-1}\otimes p_{r_2}\otimes I_{r_2+1}\cdots \otimes I_n =
\sem{p_{r_1}p_{r_2}}_m$ if $r_1 < r_2$. 

One primary concern of Pauli expression syntax lies in its closedness under the unitary transformations Clifford + $T$ as claimed below. In fact, the factor $SExp$ is introduced to ensure the closedness under the $T$ gate.

\begin{theorem}[Closedness of Pauli expression under Clifford + $T$, c.f. \cite{sundaram2022hoare}]
\label{thm-pauli-closed}
For any Pauli expression $P$ defined in Eqn. \eqref{pauli-exp} and single-qubit gate $U_1\in\{X,Y,Z,H,S,T\}$ acts on $q_i$ or two-qubit gate $U_2\in\{CNOT, CZ, iSWAP\}$ acts on $q_iq_j$, there exists another Pauli expression $Q\in PExp$, such that for all $m\in\mathtt{CMem}$, 
$\sem{Q}_m = U_{1i}^\dag\sem{P}_m U_{1i}$ or $\sem{Q}_m = U_{2ij}^\dag\sem{P}_m U_{2ij}$.
\end{theorem}

\subsection{Assertion language}\label{assert}
We further define the assertion language for QEC codes by adopting Boolean and Pauli expressions as atomic propositions. Pauli expressions characterize the stabilizer group and the subspaces stabilized by it, while Boolean expressions are employed to represent error properties.

\begin{definition}[Syntax of assertion language]
\label{assert-lang}
    \begin{align}
    AExp: \quad A \Coloneqq{} &  b \in BExp \mid P \in PExp \mid \neg A \mid A \wedge A \mid A \vee A \mid  A \Rightarrow A.
    \end{align}
\end{definition}
We interpret the assertion $A \in AExp $ as a map $\sem{A}: \mathtt{CMem} \rightarrow \mathcal{S}(\mathcal{H})$, where $\mathtt{CMem}$ is the set of classical states, $\mathcal{S}(\mathcal{H})$ is the set of subspaces in global Hilbert space $\mathcal{H}$. Formally, we define its semantics as:
\begin{align*}
    &\sem{b}_m \triangleq \left\{\begin{array}{ll}
        I_{\mathcal{H}} & \sem{b}_m = \mathbf{true} \\
        0_{\mathcal{H}} & \sem{b}_m = \mathbf{false}
    \end{array} \right., \quad
    \sem{P}_m \triangleq \mathrm{span}\{|\psi\rangle: \sem{P}_m|\psi\rangle = |\psi\rangle\}, \quad
    \sem{\neg A}_m \triangleq \sem{A}^\bot_m,\\
    &\sem{A_1 \wedge A_2}_m \triangleq \sem{A_1}_m\wedge\sem{A_2}_m,\quad
    \sem{A_1 \vee A_2}_m \triangleq \sem{A_1}_m\vee\sem{A_2}_m,\quad
    \sem{A_1 \Rightarrow A_2}_m \triangleq \sem{A_1}_m\rightsquigarrow\sem{A_2}_m
\end{align*}
Boolean expression is embedded as null space or full space depending on its boolean semantics. Pauli expression is interpreted as its +1-eigenspace (aka codespace), intuitively, this is the subspace of states that are stabilized by it. It is slightly ambiguous to use $\sem{P}$ for both semantics of $PExp$ and $AExp$, while it can be recognized from the context if $\sem{P}_m$ refers to operator ($PExp$) or subspace ($AExp$).
For the rest of connectives, $\sem{\cdot}$ is a point-wise extension of quantum logic, i.e., ${}^\bot$ as orthocomplement, $\wedge$ as intersection, $\vee$ as span of union, $\rightsquigarrow$ as Sasaki implication of subspaces, i.e., $a\rightsquigarrow b\triangleq \neg a\vee (a\wedge b)$. Sasaki implication degenerates to classical implication whenever $a$ and $b$ commute, and thus it is consistent with boolean expression, e.g., $\sem{b_1\rightarrow b_2} = \sem{b_1\Rightarrow b_2}$ where $\rightarrow$ is the boolean implication. See Appendix \ref{review-subspace} for more details.

\subsection{Why Birkhoff-von Neumann quantum logic as base logic?}
In this section, we will discuss the advantages of choosing the projection-based (Birkhoff-von Neumann) quantum logic as the base logic to verify QEC programs. 

    \paragraph{Quantum logic vs. Classical logic}
    A key difference is the interpretation of $\vee$, which is particularly useful for backward reasoning about $\mathbf{if}$-branches, as shown by rule (If) in Fig. \ref{infer-rule} that aligns with its counterpart in classical Hoare logic.
    However, interpreting $\vee$ as the classical disjunction is barely applicable for backward reasoning about measurement-based $\mathbf{if}$-branches, as illustrated below.
    
\vspace{-0.06cm}

    \begin{example}[Failure of backward reasoning about $\mathbf{if}$-branches with classical disjunction]
    \label{exam-failure-disjunction}
    Consider a fragment of QEC program $S\equiv \qassign{\mathbf{meas}[Z_2]}{b}; \qqif{b}{\qut{X}{q_2}}{\mathbf{skip}}$, which first detects possible errors by performing a computational measurement\footnote{Note that $\mathtt{P}_{\sem{Z_2}_m} = |0\>_{q_2}\<0|$ and $\mathtt{P}_{\sem{Z_2}_m^\bot} = |1\>_{q_2}\<1|$, so $\qassign{\mathbf{meas}[Z_2]}{b}$ represents the computational measurement on $q_2$ and assign the output to $b$.} on $q_2$ and then corrects the error by flipping $q_2$ if it is detected.
    It can be verified that the output state is stabilized by $X_1\wedge Z_2$ (i.e., in state $\ket{+0}_{q_1q_2}$) after executing $S$, if the input state is stabilized by $X_1$ (i.e., in state $\ket{+}_{q_1}\ket{\psi}_{q_2}$ for arbitrary $|\psi\>$). This fact can be formalized by correctness formula 
    \begin{equation}
    \label{eqn:motivating-example}
    \{X_1\} \ \qassign{\mathbf{meas}[Z_2]}{b}; \qqif{b}{\qut{X}{q_2}}{\mathbf{skip}}\ \{X_1\wedge Z_2\}.
    \end{equation}
    When deriving the precondition with rule (If) where $\vee$ is interpreted as classical disjunction, one can obtain the semantics of precondition as $\sem{A_0\vee A_1}' = \sem{A_0}\cup \sem{A_1} = \{\ket{+0}_{q_1q_2},\ket{+1}_{q_1q_2}\}$, where $A_0\triangleq X_1\wedge Z_2$ and $A_1\triangleq X_1\wedge-Z_2$.
    This semantics of precondition is valid but far from fully characterizing all valid inputs mentioned earlier, i.e., states of the form $\ket{+}_{q_1}\ket{\psi}_{q_2}$ for arbitrary $|\psi\>$.
    \end{example}

\vspace{-0.06cm}

    Quantum logic naturally addresses this failure, since the semantics of precondition is exactly the set of all valid input states: $\sem{A_0\vee A_1}=\mathrm{span}\{\sem{A_0}\vee\sem{A_1}\}=\{\alpha \ket{+0}_{q_1q_2}+\ket{+1}_{q_1q_2}:\alpha,\beta\in\mathbb{C}\} = \sem{X_1}$.
    As Theorem A.11 suggested, the rules (If) and (Meas) maintain the universality and completeness of reasoning about broader QEC codes.

\vspace{-0.06cm}

    \paragraph{Projection-based vs. satisfaction-based approach}
    Although quantum logic offers richer algebraic structures, it is limited in expressiveness compared to observable-based satisfaction approaches~\cite{qwp,ying2012floyd} and effect algebras~\cite{effectalgebra1,effectalgebra2}: it cannot express or reason about the probabilistic properties of programs. However, this limitation is tolerable for reasoning about QEC codes. On one hand, errors in QEC codes are discretized as Pauli errors and do not directly require modeling the probability. On the other hand, a QEC code can perfectly correct discrete errors with non-probabilistic constraints. Therefore, representing and reasoning about the probabilistic attributes of QEC codes is unnecessary.

\subsection{Satisfaction Relation and Entailment}

In this section, we first review the representation of program states and then define the satisfaction relation, which specifies when the program states meet the truth condition of the assertion under a given interpretation.

\vspace{-0.06cm}

    \paragraph{Quantum states as density operators.}
    The quantum system after a measurement is generally an ensemble of pure state $\{p_i,|\psi_i\>\}$, i.e., the system is in $|\psi_i\>$ with probability $p_i$. It is more convenient to express quantum states as partial density operators instead of pure states~\cite{nielsen2010quantum}. Formally, we write $\rho\triangleq \sum_ip_i|\psi_i\>\<\psi_i|\in\cD(\cH)$, where $\<\psi_i|$ is the dual state, i.e., the conjugate transpose of $|\psi_i\>$.

\vspace{-0.06cm}

    \paragraph{Classical-quantum states.} We follow~\cite{feng2021quantum} to define the program state in our language as a classical-quantum state $\mu : \mathtt{CMem}\rightarrow\mathcal{D}(\mathcal{H})$, which is a map from classical states to partial density operators over the whole quantum system.
    In particular, the singleton state, i.e., the classical state $m$ associated with quantum state $\rho$, is denoted by $(m,\rho)$. 

\vspace{-0.06cm}

    \paragraph{Satisfaction relation.}
    A one-to-one correspondence exists between projective operators and subspace, i.e., $X = \{|\psi\>: \mathtt{P}_X |\psi\> = |\psi\>\}$. Therefore, there is a standard way to define the satisfaction relation in projection-based approach~\cite{zhou2019applied,unruh2019quantum}, i.e., a quantum state $\rho$ satisfies a subspace $X$, written $\rho\models X$, if and only if $\mathrm{supp}(\rho)\subseteq X$, or equivalently, $\mathtt{P}_X\rho \mathtt{P}_X = \rho$ (or $\mathtt{P}_X\rho = \rho$) where $\mathtt{P}_X$ is the corresponding projective operation of $X$. The satisfaction relation of classical-quantum states is a point-wise lifting: 
    
    \begin{definition}[Satisfaction relation]
    Given a classical-quantum state $\mu$ and an assertion $A\in AExp$, the satisfaction relation is defined as: $\mu \models A$ iff for all $m \in \mathtt{CMem}$, $\mu(m)\models \sem{A}_m$.
    \end{definition}
    The satisfaction relation faithfully characterizes the relationship of stabilizer generators and their stabilizer states, i.e., for a Pauli expression $P$, $|\psi\>\<\psi|\models P$ iff $|\psi\>$ is a stabilizer state of $\sem{P}_m$ for any $m\in \mathtt{CMem}$.
    We further define the entailment between two assertions: 
    \begin{definition}[Entailment] For $A,B\in AExp$, the entailment and logical equivalence are: 
    \begin{enumerate}
        \item $A$ entails $B$, denoted by $A\models B$, if for all classical-quantum states $\mu$, $\mu\models A$ implies $\mu\models B$.
        \item $A$ and $B$ are logically equivalent, denoted by $A \eqmodels B$, if $A\models B$ and $B\models A$.
    \end{enumerate}
    \end{definition}
    The entailment relation is also a point-wise lifting of the inclusion of subspaces, i.e., $A\models B$ iff for all $m$, $\sem{A}_m \subseteq \sem{B}_m$. As a consequence, the proof systems of quantum logic remain sound if its entailment is defined by inclusion, e.g., a Hilbert-style proof system for $AExp$ is presented in Appendix \ref{proof-sys-assert}.
    In the (consequence) rule (Fig. \ref{infer-rule}) , strengthening the precondition and weakening the postcondition are defined as entailment relations of assertions. Therefore, entailment serves as a basis for verification conditions, which are established according to the consequence rule. 

To conclude this section, we point out that the introduction of our assertion language enables us to leverage the following observation in efficient QEC verification: 

\begin{observation}
Verifying the correctness of quantum programs requires verification for all states within the state space. By introducing phase factor $(-1)^b$ to Pauli expressions, we can circumvent the need to verify each state individually. Consider a QEC code in which a logical state $\vert b_1\cdots b_k\rangle_L$ is stabilized by the set of generators and logical operators $\langle g_1, \cdots, g_{n-k},(-1)^{b_1}\bar{Z}_1,\cdots,(-1)^{b_k}\bar{Z}_k\rangle$. We can simultaneously verify the correctness for all logical states from the set  $\{\vert b_1\cdots b_k\rangle_L :  b_1,\cdots,b_k\in\{0,1\}\}$, without introducing exponentially many assertions. 
\end{observation}

\section{A Programming Language for QEC Codes and Its Logic}\label{prog-logic}
In this section, we introduce our programming language and the program logic specifically designed  for QEC programs.
\subsection{Syntax and Semantics}
The set of program commands $\qprog$ is defined as follows:
    \begin{align*}
        \qprog:\quad 
        S \Coloneqq{} & \qskip \mid \qinit{q_i} \mid \qut{U_1}{q_i} \mid \qut{U_2}{q_iq_j}  
        &\qquad &\mbox{where:} 
        \\
        & \qassign{e}{x} \mid  
        \qassign{\mathbf{meas}[P]}{x} 
         \mid S \fatsemi S && U_1 \in \{X, Y, Z, H, S, T\}\\            
        & \qqif{b}{S}{S} \mid \qqwhile{b}{S} && U_2 \in \{CNOT, CZ, iSWAP\}
    \end{align*}
where $\qskip$ denotes the empty program, and $\qinit{q_i}$ resets the $i$-th qubit to ground state $|0\>$. A restrictive but universal gate set is considered for unitary transformation, with single qubit gates from $\{X,Y,Z,H,S,T\}$ and two-qubit gates from $\{CNOT, CZ, iSWAP\}$, where $i$ and $j$, as the indexes of unitaries, are constants and $i\neq j$ for two-qubit gates. $\qassign{e}{x}$ is the classical assignment. In quantum measurement $\qassign{\mathbf{meas}[P]}{x}$, $P\in PExp$ is a Pauli expression which defines a projective measurement $\{M_0 = \mathtt{P}_{\sem{P}_m}, M_1 = \mathtt{P}_{\sem{P}_m^\bot}\}$; after performing the measurement, the outcome is stored in classical variable $x$. $S \fatsemi S$ is the sequential composition of programs. In if/loop commands, guard $b\in BExp$ is a Boolean expression, and the execution branch is determined by its value $\sem{b}_m$.

Our language is a subset of languages considered in~\cite{feng2021quantum}, and we follow the same theory of defining operational and denotational semantics. In detail, a classical-quantum configuration is a pair $\langle S, (m,\rho)\rangle$, where $S$ is the program that remains to be executed with extra symbol $\downarrow$ for termination, and $(m,\rho)$ the current singleton states of the classical memory and quantum system. The transition rules for each construct are presented in Fig. \ref{op-semantic}. We can further define the induced denotational semantics $\sem{S} : (\mathtt{CMem}\times\mathcal{D}(\mathcal{H}))\rightarrow(\mathtt{CMem}\rightarrow\mathcal{D}(\mathcal{H}))$, which is a mapping from singleton states to classical-quantum states~\cite{feng2021quantum}.
We review the technical details in Appendix \ref{sec:app-denotational-semantics}.

\begin{figure}
\small
\flushleft
    \begin{align*} 
    & \mbox{(Skip)} \ \langle \mathbf{skip},(m,\rho)\rangle \rightarrow \langle \downarrow,(m,\rho)\rangle &
    \mbox{(Init)} \ \langle \qassign{\vert 0\rangle}{q_i}, (m,\rho)\rangle \rightarrow \langle \downarrow, (m,
    \mbox{$\sum_{k = 0,1}$}\vert 0\rangle_{q_i}\langle k\vert \rho\vert k\rangle_{q_i}\langle 0\vert)\rangle  \\
    & \mbox{(Unit1)} \ \langle \qut{U}{q_i},(m, \rho)\rangle \rightarrow  \langle \downarrow, (m, U_{q_{i}}\rho U_{q_{i}}^{\dag})\rangle &
    \mbox{(Unit2)} \ \langle \qut{U}{q_iq_j},(m, \rho)\rangle \rightarrow  \langle \downarrow, (m, U_{q_{i,j}}\rho U_{q_{i,j}}^{\dag})\rangle \\
    & \mbox{(Assign)}\ \langle \qassign{e}{x},(m, \rho)\rangle \rightarrow  \langle \downarrow, (m[\sem{e}_m/x],\rho)\rangle &
    \mbox{(Meas)} \ \inference{M_{0} = \mathtt{P}_{\sem{P}_m}, M_{1} = \mathtt{P}_{\sem{P}_m^\bot}}{\hspace{-0.2cm}\langle \qassign{\mathbf{meas}[P]}{x}, (m,\rho)\rangle \rightarrow \langle \downarrow,(m[j/x],M_j\rho M_j^{\dag})\rangle\hspace{-0.2cm}} \\
    & \mbox{(Seq)} \  \inference{\langle S_1,(m,\rho) \rangle \rightarrow \langle S_1',(m',\rho')\rangle}{\hspace{-0.2cm}\langle S_1\fatsemi S_2, (m,\rho)\rangle \rightarrow \langle S_1'\fatsemi S_2,(m', \rho')\rangle\hspace{-0.2cm}} &
    \mbox{(If-F)} \ \inference{\sem{b}_m = \mathbf{false}}{\hspace{-0.2cm}\langle \qqif{b}{S_1}{S_0},(m,\rho)
    \rangle \rightarrow \langle S_0, (m,\rho)\rangle\hspace{-0.2cm}} \\ 
    \medskip
    & \mbox{(While-F)}\ \inference{\sem{b}_m = \mathbf{false}}{\hspace{-0.2cm}\langle\qqwhile{b}{S},(m,\rho)\rangle \rightarrow \langle \downarrow,(m,\rho)\rangle\hspace{-0.2cm}} \hspace{-1cm} & 
    \mbox{(If-T)} \ \inference{\sem{b}_m = \mathbf{true}}{\hspace{-0.2cm}\langle \qqif{b}{S_1}{S_0},(m,\rho)\rangle \rightarrow \langle S_1, (m,\rho)\rangle\hspace{-0.2cm}}  \\
    & \mbox{(While-T)} \
    \inference{\sem{b}_m = \mathbf{true}}{\langle\qqwhile{b}{S},(m,\rho)\rangle \rightarrow \langle S\fatsemi\qqwhile{b}{S},(m,\rho)\rangle}\hspace{-6cm} &
    \end{align*}
\vspace{-0.35cm}
\caption{Operational semantics for QEC programs.}
\label{op-semantic}
\end{figure}

    \paragraph{Expressiveness of the programming language}
    Our programming language supports Clifford + T gate set and Pauli measurements. Therefore, it is capable of expressing all possible quantum operations, in an approximate manner. The claim of expressiveness can be proved by the following observations:
    \begin{enumerate}
        \item Clifford + $T$ is a universal gate set~\cite{nielsen2010quantum}. Thus, according to the Solovay-Kitaev theorem, any unitary $U$ can be approximated within error $\epsilon$ using $\Theta(\log^c(1/\epsilon))$ gates from this set, where $c$ is a constant whose value depends on the proof.
        \item Measurement in any computational basis $|m\rangle = |a_1a_2\cdots a_n\rangle$ is performed by the projector $\mathtt{P}_m = \frac{\Pi_{i=1}^{n}(I + (-1)^{a_i}Z_i)}{2^n}$, which can be expressed using our measurement statements $\qmeas{x}{(-1)^{a_i}Z_i}$. Further, projective measurements along with unitary operations are sufficient to implement any POVM measurement.
    \end{enumerate}

\subsection{Correctness formula and proof system} \label{corr-proof-sys}

\begin{definition}[Correctness formula]
\label{correct-formula}
    The correctness formula for QEC programs is defined by the Hoare triple $\{A\} S \{B\}$, where $S\in Prog$ is a QEC program, $A,B\in AExp$ are the pre- and post-conditions.
    A formula $\{A\} S \{B\}$ is valid in the sense of partial correctness, written as $\models \{ A \} S \{ B \}$, if for any singleton state $(m,\rho)$: 
    $(m,\rho) \models A$ implies $\sem{S}(m,\rho) \models B$.
\end{definition}

The proof system of QEC program is presented in Fig. \ref{infer-rule}. Most of the inference rules are directly inspired from~\cite{ying2012floyd, zhou2019applied, feng2021quantum}.  
We use $A[e/x]$ (or $A[e_1/x_1,e_2/x_2,\cdots]$) to denote the (simultaneous) substitution of variable $x$ or constant constructor $x\in \{X_r,Y_r,Z_r\}$ with expression $e$ in assertion $A$. 
Based on the syntax of our assertion language and program constructors, we specifically design the following rules:
\begin{itemize}
    \item Rule (Init) for initialization. Previous works~\cite{ying2012floyd,feng2021quantum} do not present syntax for assertion language and give the precondition based on the calculation of semantics, which, however, cannot be directly expressed in $AExp$. We derive the rule (Init) from the fact that initialization can be implemented by a computational measurement followed by a conditional $X$ gate. As shown in the next section, the precondition is indeed the weakest precondition and semantically equivalent to the one proposed in~\cite{zhou2019applied}.
    \item Rules for unitary transformation. We provide the rules for Clifford + $T$ gates, controlled-Z ($CZ$) gate, as well as $iSWAP$ gate, which are easily implemented in superconducting quantum computers. It is interesting to notice that, even for two-qubit unitary gates, the pre-conditions can still be written as the substitution of elementary Pauli expressions. 
\end{itemize}

\paragraph{Reasoning about Pauli errors}
To model the possible errors occurring in the QEC program, we further introduce a syntax sugar $\qut{U}{[b]q_i}$ for `$\qqif{b}{\qut{U}{q_i}}{\qskip}$' command, which means if the guard $b$ is true then apply Pauli error $U\in\{X,Y,Z\}$ on $q$, otherwise skip. The corresponding derived rules are:
\begin{equation*}
    \hspace{-0.3cm}
    \begin{array}{ll}
       \big\{ A[(-1)^b Y_i/Y_i, (-1)^{b}Z_i/Z_i]\big\} \ \qut{X}{[b]q_i} \ \{A \} & \big\{ A[(-1)^{b}X_i/X_i, (-1)^{b}Z_i/Z_i]\big\} \ \qut{Y}{[b]q_i} \ \big\{A \big\} \\[0.1cm]
       \big\{ A[(-1)^{b}X_i/X_i, (-1)^bY_i/Y_i]\big\} \ \qut{Z}{[b]q_i} \ \big\{A \big\}. & 
    \end{array}
\end{equation*}

\begin{example}[Derivation of the precondition using the proof system]
    Consider a fragment of QEC program which describes the error correction stage of 3-qubit repetition code: $\qqfor{1\dots 3}{\qut{X}{[x_i]q_i}}$. This program corrects possible $X$ errors indicated by $x_i$. Starting from the post-condition $Z_1Z_2\wedge Z_2Z_3 \wedge (-1)^b Z_1$, we derive the weakest pre-condition for this program: 
    \begin{align*}
        \big\{ Z_1Z_2 \wedge (-1)^{x_3} Z_2Z_3 \wedge (-1)^b Z_1\big\}\  & \qut{X}{[x_3]q_3} \ \big\{ Z_1Z_2 \wedge Z_2Z_3 \wedge (-1)^b Z_1\big\} \\
        \big\{ (-1)^{x_2}Z_1Z_2 \wedge (-1)^{x_3+x_2} Z_2Z_3 \wedge (-1)^b Z_1\big\}\  &\qut{X}{[x_2]q_2} \ \big\{ Z_1Z_2 \wedge (-1)^{x_3} Z_2Z_3 \wedge (-1)^b Z_1\big\} \\
        \big\{ \mbox{\small$(-1)^{x_2+x_1} Z_1Z_2 \wedge (-1)^{x_3 + x_2} Z_2Z_3 \wedge (-1)^{b+x_1} Z_1$} \big\}\  &\qut{X}{[x_1]q_1} \ \big\{ \mbox{\small$(-1)^{x_2}Z_1Z_2 \wedge (-1)^{x_3+x_2}Z_2Z_3 \wedge (-1)^b Z_1$}\big\} 
    \end{align*}
    We break down the syntax sugar as a sequence of subprograms and use the inference rules for Pauli errors to derive the weakest pre-condition.
\end{example}

\begin{figure}[t]
\centering
\small
\begin{align*}
    &\mbox{(Skip)} \ \vdash\{A\} \ \qskip \ \{A\} \hspace{-5cm} 
    & \mbox{(Init)} \ \vdash\{ (Z_i \wedge A) \vee (-Z_i \wedge A[-Y_i/Y_i, -Z_i/Z_i]) \} \ \qassign{\vert 0\rangle}{q_i} \  \{A\} \\[0.1cm]
    &\mbox{(Assign)} \ \vdash\{A[e/x]\} \qassign{e}{x} \ \{ A \} \hspace{-5cm} 
    & \mbox{(Meas)} \ \vdash\{ (P \wedge A[0/x])\vee (\neg P \wedge A[1/x])\} \ \qassign{\mathbf{meas}[P]}{x} \ \{A\} \\[-0.1cm]
   \cline{1-3} \\[-0.45cm]
   &\mbox{(U-X)} \ \vdash\{ A[-Y_i/Y_i,-Z_i/Z_i]\} \ \qut{X}{q_i} \ \{A \} \hspace{-5cm} &
    \mbox{(U-Y)} \ \vdash\{ A[-X_i/X_i, -Z_i/Z_i]\} \ \qut{Y}{q_i} \ \{A \}\\[0.1cm]
   &\mbox{(U-Z)} \ \vdash\{ A[-X_i/X_i, -Y_i/Y_i]\} \ \qut{Z}{q_i} \ \{A \} \hspace{-5cm} &
    \mbox{(U-H)} \ \vdash\{A[Z_i/X_i, -Y_i/Y_i, X_i/Z_i]\} \ \qut{H}{q_i} \ \{A\} \\[0.1cm]
   &\mbox{(U-S)} \ \vdash\{A[-Y_i/X_i, X_i/Y_i]\} \ \qut{S}{q_i} \ \{A\} \hspace{-5cm} &
    \mbox{(U-T)} \ \vdash\{A[\mbox{$\frac{1}{\sqrt{2}}$}(X_i - Y_i) /X_i, \mbox{$\frac{1}{\sqrt{2}}$}(X_i + Y_i)/Y_i] \ \qut{T}{q_i} \ \{ A \} \\[0.1cm]
   &\mbox{(U-CNOT)}  \ \vdash\{A[X_iX_j/X_i, Y_iX_j/Y_i, Z_iY_j/Y_j, Z_iZ_j/Z_j]\} \ \qut{CNOT}{q_i q_j} \ \{ A \} \hspace{-5.5cm} & \\[0.1cm]
   &\mbox{(U-CZ)} \ \vdash\{ A[X_iZ_j/X_i, Y_iZ_j/Y_i, Z_iX_j/X_j, Z_iY_j/Y_j] \} \ \qut{CZ}{q_iq_j} \ \{ A \} \hspace{-5.5cm} & \\[0.1cm]
   &\mbox{(U-iSWAP)} \ \vdash\{A[Z_iY_j/X_i, -Z_iX_j/Y_i, Z_j/Z_i, Y_iZ_j/X_j, -X_iZ_j/Y_j, Z_i/Z_j] \} \ \qut{iSWAP}{q_iq_j} \ \{ A \} \hspace{-7cm} & \\[-0.1cm]
   \cline{1-3} \\[-0.45cm]
   & \mbox{(Seq)} \ \inference{\vdash\{A\}S_1\{B\} \quad \vdash\{B\} S_2\{C\} }{\vdash\{A\} S_1\fatsemi S_2 \{C\}} \hspace{-5cm} 
   & \mbox{(If)}  \ \inference{ \vdash\{A_0\} S_0 \{B\}\quad \vdash\{ A_1\}S_1 \{B\}  }{\vdash\{(\neg b \wedge A_0) \vee (b\wedge A_1) \} \ \qqif{b}{S_1}{S_0} \ \{B\} } \\[0.1cm]
   &\mbox{(While)} \ \inference{\vdash\{b \wedge A\} S \{ A \} }{\vdash\{ A \} \ \qqwhile{b}{S} \ \{\neg b \wedge B\}} \hspace{-5cm} 
   & \mbox{(Con)} \ \inference{ A \models A' \quad \vdash\{A'\} S \{B'\} \quad B^{\prime}\models B}{ \vdash\{A\} S \{A\}} 
\end{align*}
\caption{Inference rules for reasoning about QEC programs. For simplicity, we write $-P$ for $(-1)^\mathbf{true}P\in PExp$, write $P_1-P_2$ for $P_1+(-1)^\mathbf{true}P_2\in PExp$, where $P, P_1, P_2\in PExp$, and write $\frac{1}{\sqrt{2}}$ for $\frac{\sqrt{2}}{2^1}\in SExp$. }
\label{infer-rule}
\end{figure}

\subsection{Soundness theorem}
In this subsection, we present the soundness of our proof system and sketch the proofs.
\begin{theorem}[Soundness]
\label{thm-sound}
The proof system presented in Fig. \ref{infer-rule} is sound for partial correctness; that is, 
for any $A,B\in AExp$ and $S\in Prog$, $\vdash \{A\} S \{ B\}$ implies $\models \{A\} S \{ B\}$.
\end{theorem}
The soundness theorem can be proved in two steps. First of all, we provide the rigorous definition of the weakest liberal precondition $wlp.S.f_B$ for any program $S\in \qprog$ and mapping $f_B:\mathtt{CMem} \rightarrow \mathcal{S}(\mathcal{H})$ and prove the correctness of this definition. Subsequently, we use structural induction to prove that for any $A,B \in AExp$ and $S\in Prog$ such that $\vdash \{A\} S \{ B\}$, $\sem{A} \models wlp.S.\sem{B}$. Proofs are discussed in detail in Appendix A.7.

\section{Verification Framework and a Case Study}
\label{sec:VC-framework}

    Now we are ready to assemble assertion logic and program logic presented in the previous two section into a framework of QEC verification. 
    
    \subsection{Verification Conditions}
    \label{verify-cond}
        As Theorem A.11 suggests, all rules except for (While) and (Con) give the weakest liberal precondition with respect to the given postconditions. Then the standard procedure like the weakest precondition calculus can be used to verify any correctness formula $\{ A \} S \{ B \} $, as discussed in~\cite{ying2024foundations}:
        \begin{enumerate}
            \item Obtain the expected precondition $A^\prime$ in $\{A^\prime\} S \{B\}$ by applying inference rules of the program logic backwards.
            \item Generate and prove the \textit{verification condition} (VC) $A\models A^\prime$ using the assertion logic.
        \end{enumerate}
        Dealing with VC requires additional efforts, particularly in the presence of non-commuting pairs of Pauli expressions. 
        However for QEC programs, there exists a general form of verification condition, which can be derived from the correctness formula:

        \begin{definition}[Correctness formula for QEC programs]
        Consider a program $S = \textbf{Corr}(E, U)$, which is generalized from the QEC program in Table \ref{tab:prog}. It operates on a stabilizer code with a minimal generating set $\{g_1, \cdots, g_{n-k}, \bar{L}_{n-k+1}, \cdots, \bar{L}_n\}$ containing $n$ independent and commuting Pauli expressions. The correctness formula of this program can be expressed as follows:
        \begin{equation}
        \bigg\{ \bigwedge_i g_i \wedge \bigwedge_j \bar{L}_j\bigg\}  \ S \ \bigg\{ \bigwedge_i g_i \wedge \bigwedge_j \bar{U}\bar{L}_j\bar{U}^{\dag} \bigg\}
        \label{def:corr-qec}
        \end{equation}
        \end{definition}
        
        The verification condition to be proven is derived from this correctness formula with the aid of inference rules, as demonstrated below\footnote{Here, we assume the error in the correction step is always Pauli errors; otherwise, two verification conditions of the form Eqn. (\ref{wp-stabilizer}) are generated that separately deal with error before measurement and error in correction step.}: 
        \begin{equation}
         \bigg( \bigwedge_{i}g_i \wedge \bigwedge_{j} \bar{L}_j \bigg) \ \wedge P_c \models \bigvee_{\bs\in \{0,1\}^{n-k}} \bigg(
         \bigwedge_{i} (-1)^{r_i(\bs) + h_i(\be)} g_i^\prime \wedge
         \bigwedge_{j} (-1)^{r_j(\bs)+h_{j}(\be)}\bar{L}_j^{\prime}
         \bigg).
        \label{wp-stabilizer}
        \end{equation}

        In Eqn. \eqref{wp-stabilizer}, $P_c$ represents a classical assertion for errors, $i, j$ range over $\{1,\cdots,n-k\}$, $\{n-k+1,\cdots,n\}$ respectively, The vector $\bs$ encapsulates all possible measurement outcomes (syndromes) and $\be$ represents the error configuration. The semantics of $g_i, g_i', \bar{L}_j, \bar{L}_j'$ are normal operators. The terms $r_i(\bs), r_j(\bs)$ denote the sum of all corrections effective for the corresponding operators, while $h_i(\be),h_j(\be)$ account for the total error effects on the operators caused by the injected errors. The details of derivation are provided in Appendix \ref{eff-def-app}.

        Let us consider how to prove Eqn. (\ref{wp-stabilizer}) in the following three cases:
        \begin{enumerate}
        \item $\{g_i'\}\subseteq \{g_i\}$ and $\{\bar{L}_j'\}\subseteq \{\bar{L}_j\}$.
        The entailment is then equivalent to check $P_c\models\bigvee_{\bs}\big(\bigwedge_i (r_i(\bs)+h_i(\be) = 0) \wedge \bigwedge_j (r_j(\bs)+h_j(\be) = 0)\big)$, which can be proved directly by SMT solvers.
        \item All $g_i, g_i', \bar{L}_j, \bar{L}_j'$ commute with each other. Since $\{g_i, \bar{L}_j \}$ is a minimal generating set, any $g_i'$ or $\bar{L}_j'$ can be written as the product of $\{g_i, \bar{L}_j \}$ up to a phase $\pm 1$, e.g., $(-1)^{\alpha_i}g_i' = \prod_{i\in \mathcal{I}_{i'}}g_i\prod_{j\in \mathcal{J}_{j'}}\bar{L}_j$,
        $(-1)^{\alpha_j}\bar{L}_j' = \prod_{i\in \mathcal{I}_{i'}}g_i\prod_{j\in \mathcal{J}_{j'}}\bar{L}_j$, so the entailment is equivalent to check $P_c\models\bigvee_{\bs}\big(\bigwedge_i (r_i(\bs)+h_i(\be) = \alpha_i) \wedge \bigwedge_j (r_j(\bs)+h_j(\be) = \alpha_j)\big)$.
        \item There exist non-commuting pairs\footnote{We assume no error happens in the correction step; otherwise, we deal them in two separate VCs.}. 
        We consider the case that the total errors are less than the code distance; furthermore, $g_i'$ is ordered such that $g_i' = Ug_iU^\dag$ for some unitary $U$, which can be easily achieved by preserving the order of subterms during the annotation step (1).
        
        The key idea to address this issue involves eliminating all non-commuting terms on the right-hand side (RHS) and identifying a form that is logically equivalent to the RHS. We briefly discuss the steps of how to eliminate the non-commuting terms, as outlined below:
        \begin{enumerate}
        \item Find the set $\mathcal{G}  \subseteq \{g_i'\}$ such that any element $g_i'\in \mathcal{G}$ differs from $g_i$ up to a phase; Find the set
        $\mathcal{L}\subseteq\{\bar{L}_j'\}$ such that $\bar{L}_j'$ differs from $\bar{L}_j$ up to a phase.
        \item Update $\mathcal{G}$ and $\mathcal{L}$ by multiplying some
        $g_i'\in \mathcal{G}$ onto those elements, until $\mathcal{L}$ is
        empty and any $g_i' \in \mathcal{G}$ differs from $g_i$ in only one qubit. 
        \item Replace those $g_i'$ with $g_i$, and check if the phases of the remaining items are the same for all $2^k$ terms. If so, this
        problem can be reduced to the commuting case, since we can successfully
        use $(P\wedge Q) \vee (\neg P \wedge Q) = Q$ ($P$ and $Q$ commute with each other) to eliminate all non-commuting elements.
         \end{enumerate} 
        To illustrate how our ideas work, we provide an concrete example in Section \ref{subse-nonpauli}, which illustrates how to correct a single $T$ error in the Steane code.
    \end{enumerate}

\textbf{Soundness of the above methods.}

After proposing the methods to handle the verification condition (VC), we now discuss the soundness of our methods case by case:

$\bullet$ \textit{Commuting case.} If all $g_i, g_i', \bar{L}_j, \bar{L}'_j$ commute with each other, then the equivalence of the VC proposed in case (2) and Eqn. \eqref{wp-stabilizer} can be guaranteed by the following proposition: 

\begin{proposition}
Given a verification condition of the form:
\begin{equation}
\Big((-1)^{b_1} P_1 \wedge \dots \wedge (-1)^{b_n} P_n \Big) \wedge P_c \models  \bigvee_{\bm{s}} \Big( (-1)^{b_1'} P_1^\prime\wedge \dots \wedge (-1)^{b_n'} P_n^\prime\Big)
\label{vc-expr-app}
\end{equation}
where $\left\{(-1)^{b_1}P_1, \dots, (-1)^{b_n}P_n\right\}$, $\left\{(-1)^{b_1'}P_1^{\prime}, \dots, (-1)^{b_n'}P_n^{\prime}\right\}$ are independent and commuting generators of two stabilizer groups $S,S' \subseteq \mathcal{G}_n$, $\mathcal{G}_n$ is the n-qubit Pauli group. $S$ and $S'$ satisfy $-I \notin S, S'$. If $\{P_1, \dots, P_n, P_1', \dots, P_n'\}$ commute with each other, then:
\begin{enumerate}[label={\Roman*.}]
    \item For all $i$, there exist a unique $\alpha_i\in \{0,1\}$ and $\{i_j\}\in 2^{[n]}, s.t. (-1)^{\alpha_i}P_i' = \Pi_j P_{i_j}$.
    \item $P_c \models  \bigwedge_{i=1}^{n}(b_i' = \alpha_i + \sum_j b_{i_j})$ implies $A\models A'$, where $A, A'$ are left and right hand side of Expression {(\ref{vc-expr-app})}.
\end{enumerate}
\end{proposition}

The proof leverages the observation that any $P_i'$ which commutes with all elements in a stabilizer group $S$ can be written as products of generators of $S$~\cite{nielsen2010quantum}. We further use $P \wedge Q = QP$ to reformulate the LHS of Expression (\ref{vc-expr-app}) and generate terms that differs from the RHS only by phases. The detailed proof of this proposition is postponed to Appendix \ref{eff-verify-app}. 

$\bullet$ \textit{Non-commuting case.} The soundness of this case can be demonstrated by separately proving the soundness of step (a), (b) and step (c).
\begin{enumerate}
\item \textit{Step (a) and (b)}: Consider the check matrix $H$. If step (b) fails for some error configuration $\be$ with weight $w_{\be}\leq d-1$, then there exists a submatrix $H_{sub}$ of size $(n-k)\times w_{\be}$, with columns being the error locations. The rank of the submatrix is $r<w_{\be}$, leading to a contradiction with the definition of $d$ being the minimal weight of an undetectable error. This is because there exists another $\be'$ whose support is within that of $\be$, and $H\be' = 0$.
\item \textit{Step (c)}: The soundness is straightforward since $(P\wedge Q) \vee (\neg P \wedge Q) = Q$ whenever $P$ and $Q$ commute, which is the only formula we use to eliminate non-commuting elements.
\end{enumerate}

\subsection{Case study: Steane code (continued)} 
\label{case-steane}
To illustrate the general procedure of our verification framework, let us consider the 7-qubit Steane code presented in Section \ref{example-steane} with $Y$ and $T$ errors ($H$ errors is deferred to Appendix \ref{case2-steane-app}. 
\subsubsection{Case I: Reasoning about Pauli $Y$ errors}
We first verify the correctness of Steane code with Pauli $Y$ errors. We choose $Y$ error because its impact on stabilizer codes is equivalent to the composite effect of $X$ and $Z$ errors on the same qubit. In this scenario, the verification condition (VC) to be proved is generated from the precondition:\footnote{The notations in Eqn. \eqref{eqn:steane} may be a bit confusing, therefore we provide Table \ref{tab:value} to help explain the relationships of those notations.  For details of the derivation please refer to Appendix \ref{case1-steane-app}}.

\begin{equation}
\left\{\bigg(\sum_{i=1}^{7} e_i\leq 1\bigg) \wedge  \bigg((-1)^{b}\bar{Z} \wedge  \bigwedge_{i = 1}^{6}g_i\bigg)\right\} \models \Bigg\{\bigvee_{\bs\in\{0,1\}^6} \bigg((-1)^{b+r_7(\bs)+h_7(\be)}\bar{Z}\wedge \bigwedge_{i=1}^{6} (-1)^{r_i(\bs)+h_i(\be)} g_i\bigg) \Bigg\}.
\label{eqn:steane}
\end{equation}
No changes occur in Pauli generators $\bar{Z}$ and $g_i$, therefore according to case (1) in the proof of Eqn. \eqref{wp-stabilizer}, the verification condition is equivalent with $P_c \sqsubseteq P_c^\prime$, where $P_c = \sum_{i=1}^7 e_i \leq 1$, $P_c^\prime = \bigvee_{\bs\in\{0,1\}^6}\bigwedge_{i=1}^{7}(r_i(\bs) + h_i(\be) = 0)$.
\begin{table}
  \caption{Symbols and values appear in Eqn. \eqref{eqn:steane}}
  \label{tab:value}
  \small
  \begin{tabular}{c|c|c|c|c|c}
    \hline
    Symbols  & Values & Symbols  & Values & Symbols & Values \\
    \hline
    $r_7(\bs)$& $\sum_{i=1}^7 f_{z,i}$ & $h_7(\be)$ & $\sum_{i=1}^7 e_i$ & &  \\
    \hspace{-0.16cm}$h_1(\be)$, $h_4(\be)$\hspace{-0.10cm} & $e_1 + e_3 + e_5 + e_7 $ & \hspace{-0.08cm}$h_2(\be)$, $h_5(\be)$\hspace{-0.10cm} & $e_2 + e_3 + e_6 + e_7$  & \hspace{-0.08cm}$h_3(\be)$, $h_6(\be)$\hspace{-0.10cm} & $e_4 + e_5 + e_6 + e_7$ \\
    $r_1(\bs)$ & \hspace{-0.08cm}$f_{z,1} + f_{z,3} + f_{z,5} + f_{z,7}$\hspace{-0.10cm} & $r_2(\bs)$ & \hspace{-0.08cm}$f_{z,2} + f_{z,3} + f_{z,6} + f_{z,7}$\hspace{-0.10cm} & $r_3(\bs)$ & \hspace{-0.08cm}$f_{z,4} + f_{z,5} + f_{z,6} + f_{z,7}$\hspace{-0.16cm} \\
    $r_4(\bs)$ & \hspace{-0.08cm}$f_{x,1} + f_{x,3} + f_{x,5} + f_{x,7}$\hspace{-0.10cm} & $r_5(\bs)$ & \hspace{-0.08cm}$f_{x,2} + f_{x,3} + f_{x,6} + f_{x,7}$\hspace{-0.10cm} & $r_6(\bs)$ & \hspace{-0.08cm}$f_{x,4} + f_{x,5} + f_{x,6} + f_{x,7}$\hspace{-0.16cm} \\
  \hline
\end{tabular}
\end{table}
We can prove the VC if the minimum-weight decoder $f$ satisfies $P_f$: 
\begin{equation*}
P_f \triangleq \left(\sum_{i=1}^{7} x_i \leq \sum_{i=1}^{7} e_i\right) \bigwedge \left(\sum_{i=1}^{7} z_i \leq \sum_{i=1}^{7} e_i\right) \bigwedge \left(\bigwedge_{i=1}^{6}(r_i(\bs) = \bs_i)\right).
\end{equation*}

\noindent This $P_f$ we give describes the necessary condition of a decoder: the corrections $r_i(\bs)$ are applied to eliminate all non-zero syndromes on the stabilizers; and weight of corrections should be less than or equal to weight of errors. Alternatively, if we know that $f$ satisfies $P_f$ (e.g., the decoder is given), we can identify $P_c$ by simplifying $P_c^\prime$ without prior knowledge of $P_c$. Instead, if we are aiming to design a correct decoder $f$,
we may extract the condition $P_f$ from the requirement $P_c\sqsubseteq P_c^\prime$. 

\subsubsection{Case II: Non-Pauli $T$ Errors}
\label{subse-nonpauli}
Here we only show the processing of specific error locations $\bm{e}_{p5} = 1$, e.g., the propagated error before logical $H$, to illustrate the heuristic algorithm proposed in Section \ref{sec:VC-framework}. The general situation only makes the formula encoding more complicated but does not introdce fundamental challenges.

We consider the logical $|+\>_L$ and $|-\>_L$ state stabilized by the stabilizer generators and logical $\bar{X}$. The verification condition generated by the program should become
\footnote{
Only logical $\bar{X}$ is considered, since logical $\bar{Z}$ is an invariant at the presence of $T$ errors because $T^{\dag} Z T = Z$.
}: 
\begin{equation}
\left(\bigwedge_{i=1}^{6} g_i \right) \wedge (-1)^b\bar{X} \models \bigvee_{\bs\in\{0,1\}^6} \Bigg(\bigg(\bigwedge_{i=1}^{6}(-1)^{\bs_i}g_i^{\prime}\bigg) \wedge (-1)^{b + r(\bs)} \bar{X}'\Bigg).
\label{eqn:T-err-wp}
\end{equation}
In which $r(\bs) = \sum_{i=1}^{7}cx_i$ is the sum of X corrections, regarding the decoder as an implicit function of $\bs$. We denote the group stabilized by $g_1,\cdots,g_6,\bar{X}$ as $\cS$. 
The injected non-Pauli error $\bm{T_5}$ changes all $X_5$ to $\sq(Y_5 - X_5)$, therefore the elements in set $\{g_1',\cdots, g_6',\bar{X}'\}$ are: 
$ g_1' = \sq X_1X_3(X_5-Y_5)X_7, \  g_2' = X_2X_3X_6X_7, \ g_3' = \sq X_4(X_5-Y_5)X_6X_7 , \bar{X}' = \sq X_1X_2X_3X_4(X_5-Y_5)X_6X_7, \  g_4' = Z_1Z_3Z_5Z_7,\  g_5' = Z_2Z_3Z_6Z_7, \ g_6' = Z_4Z_5Z_6Z_7$.
\paragraph{$\bullet$\ Step I: Update $\cG$ and $\cL$.}
We obtain a subset from $\{g_1',\cdots,g_6',\bar{X}'\}$ whose elements differ from the corresponding ones in $\{g_1,\cdots g_6, \bar{X}\}$, which is $\{g_1', g_3', \bar{X}'\}$. Now pick $j_x = 1$ from this set and update $g_3'$ and $\bar{X}'$, we can obtain a generator set of $\cS'$: 
$g_1'= \sq X_1X_3(X_5-Y_5)X_7,\  g_2' = X_2X_3X_6X_7, \ g_3'' = X_1X_3X_4X_6 ,\ 
\bar{X}''=  X_2X_4X_6, \  g_4' = Z_1Z_3Z_5Z_7,\  g_5 = Z_2Z_3Z_6Z_7, \ g_6' = Z_4Z_5Z_6Z_7.$
We update $g_3, \bar{X}$ at the same time and obtain another set of generators for $\cS$: 
$\cS = \{ X_1 X_3 X_5 X_7, X_2 X_3 X_6 X_7$ , $X_1 X_3 X_4 X_6, X_2X_4X_6, Z_1Z_3Z_5Z_7, Z_2Z_3Z_6Z_7,Z_4 Z_5 Z_6Z_7 \}$.
The generator sets only differ by $g_1$ and $g_1'$.
\paragraph{$\bullet$ Step II: Remove non-commuting terms, check the phases of remaining elements. } The weakest liberal precondition on the right-hand side is now transformed into another equivalent form:
\begin{align}
\bigvee_{\bs\in\{0,1\}^6} \Bigg((-1)^{\bs_1}g_1' \wedge (-1)^{\bs_2}g_2' \wedge (-1)^{\bs_2 + \bs_3} g_3'' \wedge \bigg(\bigwedge_{i =4}^{6} (-1)^{\bs_i}g_i'\bigg) \wedge (-1)^{b+r(\bs)+\bs_1} \bar{X}''\Bigg).
\label{eqn:T-err-wp2}
\end{align}
For $P', Q$ whose elements are commute with each other, we can leverage $(P'\wedge Q)\vee (\neg P' \wedge Q) = Q$ to reduce the verification condition Eqn. \eqref{eqn:T-err-wp} to the commuting case. In this case we have $P = g_1$, $P' = g_1'$ and $Q$ being other generators, which is guaranteed by Step I. 
To prove the entailment in Eqn. \eqref{eqn:T-err-wp}, it is necessary to find two terms in Eqn. \eqref{eqn:T-err-wp2} whose phases only differ in $s_1$. Now rephrase each phase to $t_i$ and find that Eqn. \eqref{eqn:T-err-wp} has an equivalent form: 
\begin{equation}
\bigg(\bigwedge_{i=1}^{6} g_i\bigg) \wedge (-1)^b\bar{X} \models \bigvee_{\bm{t}\in\{0,1\}^7} \Bigg( (-1)^{t_1}g_1' \wedge (-1)^{t_2}g_2' \wedge (-1)^{t_3}g_3'' \wedge \bigg(\bigwedge_{i = 4}^{6} (-1)^{t_i}g_i'\bigg) \wedge (-1)^{b + t_7}\bar{X}''\Bigg).
\label{eqn:T-err-wp3}
\end{equation}
The map $\bm{t} = u(\bs)$ is $t_1 = s_1, t_2 = s_2, t_3 = s_2 + s_3, t_4 = s_4, t_5 = s_5, t_6 = s_6, t_7 = \sum_{i=1}^7 c_i + s_1 $, which comes from the multiplication in Step I. 
To prove the entailment in Eqn. (\ref{eqn:T-err-wp3}), we pick $\bm{t}$ according to step (c) in Section \ref{verify-cond} and use $\bm{t} = u(\bs)$ as constraints to check phases of the remaining items. In this case the values of $\bs_0$ and $\bs_1$ are straightforward:
$\bs_0 = (0,0,0,0,0,0)$ and $\bs_1 = (1,0,1,0,0,0)$.
Then what remains to check is whether $t_7 = \sum_{i=1}^7 cx_i + s_1 = 0$, which can be verified through the following logical formula for decoder: $H_z(\bm{cx}) = \bm{s_z} \wedge (\sum_i cx_i \leq \sum_i ex_i \leq 1) \implies \sum_{i = 1}^7 cx_i + s_1 = 0$.\footnote{The stabilizer generator $g_1$ is transformed to a $Z$-check after the logical Hadamard gate, so parity-check of $Z$ are encoded in the logical formula and the syndrome $s_1$ guides the $X$ corrections.}

\section{Tool implementation}
\label{sec:tool-implementation}
As summarized in Fig. \ref{fig:framework}, 
we implement our QEC verifiers  at two levels: a verified QEC code verifier in the Coq proof assistant~\cite{coq} for mechanized proof of scalable codes, and an automatic QEC verifier \veriq based on Python and SMT solver for small and medium-scale codes.


\paragraph{Verified QEC verifier}
Starting from first principles, we formalize the semantics of classical-quantum programs based on ~\cite{feng2021quantum}, and then build the verified prover, proving the soundness of its program logic. This rules out the possibility that the program logic itself is flawed, especially considering that it involves complex classical-quantum program semantics and counterintuitive quantum logic.
This is achieved by $\sim$4700 lines of code based on the CoqQ project~\cite{zhou2023coqq}, which offers rich theories of quantum computing and quantum logic, as well as a framework for quantum program verification.
We further demonstrate its functionality in verifying scalable QEC codes using repetition code as an example, where the size of the code, i.e., the number of physical qubits, is handled by a meta-variable in Coq.

\vspace{-0.07cm}

\paragraph{Automatic QEC verifier \veriq}
We propose \veriq, an automatic QEC code verifier implemented as a Python package. It consists of three components: 
\begin{enumerate}
    \item Correctness formula generator. This module processes the user-provided information of the given stabilizer code, such as the parity-check matrix and logical algorithms to be executed, and generates the correctness formula expressed in plain context as the verification target.
    \item Verification condition generator. This module consists of 1) a parser that converts the program, assertion, and formula context into corresponding abstract syntax trees (AST), 2) a precondition generator that annotates the program according to inference rules (as Theorem A.11 suggests, all rules except (While) and (Con) give the weakest liberal precondition), and 
    3) a VC simplifier that produces the condition to be verified with only classical variables, leveraging assertion logic and techniques proposed in Section \ref{verify-cond}.
    \item SMT checker. This component adopts Z3~\cite{de2008z3} to encode classical verification conditions into formulae of SMT-LIBv2 format, and invokes appropriate solvers according to the type of problems. We further implement a parallel SMT checking framework in our tool for enhanced performance.
\end{enumerate}
Readers can refer to Appendix \ref{details-imp} for specific details on the tool implementation.

\section{Evaluation of \veriq }
\label{sec:evaluation}

We divide the functionalities of \veriq into two modules: the first module focuses on verifying the general properties of certain QEC codes, while
the second module aims to provide alternative solutions for large QEC codes whose scales of general properties have gone beyond the upper limit of verification capability. In this case, we allow users to impose extra constraints on the error patterns.

Next, we provide the experimental results aimed at evaluating the functionality of our tool. In particular, we are interested in the performance of our tool regarding the following functionalities:
\begin{enumerate}
    \item The effectiveness and scalability when verifying the general properties for program implementations of QEC codes.
    \item The performance improvement when extra constraints of errors are provided by users. 
    \item The capability to verify the correctness of realistic QEC scenarios with regard to fault-tolerant quantum computation. 
    \item Providing a benchmark of the implementation of selected QEC codes with verified properties.
\end{enumerate}

The experiments in this section are carried out on a server with 256-core AMD(R) EPYC(TM) CPU @2.45GHz and 512G RAM, running Ubuntu 22.04 LTS. Unless otherwise specified, all verification tasks are executed using 250 cores. The maximum runtime is set to $24$ hours. 
\subsection{Verify general properties}
We begin by examining the effectiveness and scalability of our tool when verifying the general properties of QEC codes.

\paragraph{\textbf{Methodology}}
We select the rotated surface code as the candidate for evaluation, which is a variant of Kitaev's surface code~\cite{kitaev1997quantum, dennis2002topological} and has been repeatedly used as an example in Google's QEC experiments based on superconducting quantum chips~\cite{Acharya2023suppress,acharya2024quantumerrorcorrectionsurface}. As depicted in Fig. \ref{fig:surface}, a $d = 5$ rotated surface code is a $5\times 5$ lattice, with data qubits on the vertices and surfaces between the vertices representing stabilizer generators. The logical operators $\bar{X}_L$(green horizontal) and $\bar{Z}_L$ (black vertical) are also shown in the figure. Qubits are indexed from left to right and top to bottom. 

For each code distance $d = 2t + 1$, we generate the corresponding Hoare triple and verify the error conditions necessary for accurate decoding and correction, as well as for the precise detection of errors. The encoded SMT formula for accurate decoding and correction is straightforward and can be referenced in Section \ref{case-steane}:
\begin{equation}
\forall e_1, \dots, e_n, \ \exists s_1,\dots, s_{n-k}, \ \sum_{i=1}^{n} e_i \leq \Big\lfloor \frac{d-1}{2}\Big\rfloor \Rightarrow \bigvee_{\bm{s}\in\{0,1\}^{n}}\left(\bigwedge_{i=1}^{n}\left(r_i(\bs) + h_i(\be) = 0\right)\wedge P_f\right).
\label{eqn:verify}
\end{equation}
To verify the property of precise detection, the SMT formula can be simplified as the decoding condition is not an obligation:
\begin{equation}
\left(1 \leq \sum_{i=1}^{n} e_i \leq d_{t} - 1\right) \Rightarrow \left(\bigwedge_{i = k}^{n}(s_i = 0)\right) \wedge \left(\bigvee_{i=0}^{k-1}(h_i(\be) \neq 0)\right). 
\label{eqn:detect}
\end{equation}
Eqn. \eqref{eqn:detect} indicates that there exist certain error patterns with weight $\leq d_t$ such that all the syndromes are $0$ but an uncorrectable logical error occurs. We expect an \textit{unsat} result for the actual code distance $d$ and all the trials $d_{t} \leq d$.
If the SMT solver reports a \textit{sat} result with a counterexample, it reveals a logical error that is undetectable by stabilizer generators but causes a flip on logical states. In our benchmark we verify this property on some codes with distance being $2$, which are only capable of detecting errors. They are designed to realize some fault-tolerant non-Clifford gates, not to correct arbitrary single qubit errors. 

Further, our implementation supports parallelization to tackle the exponential scaling of problem instances. We split the general task into subtasks by enumerating the possible values of $e_i$ on selected qubits and delegating the remaining portion to SMT solvers. We denote $N(\text{bits})$ as the number of $e_i$ whose values have been enumerated, and $N(\text{ones})$ as the count of $e_i$ with value $1$ among those already enumerated. We design a heuristic function $ET = 2d * N(\text{ones}) + N(\text{bits})$, which serves as the termination condition for enumeration. 

Given its outstanding performance in solving formulas with quantifiers, we employ \texttt{CVC5}~\cite{DBLP:conf/tacas/BarbosaBBKLMMMN22} as the SMT solver to check the satisfiability of the logical formulas in this paper. 
\begin{figure}
\begin{minipage}{0.43\textwidth}
    \centering
    \includegraphics[width = 1\linewidth]{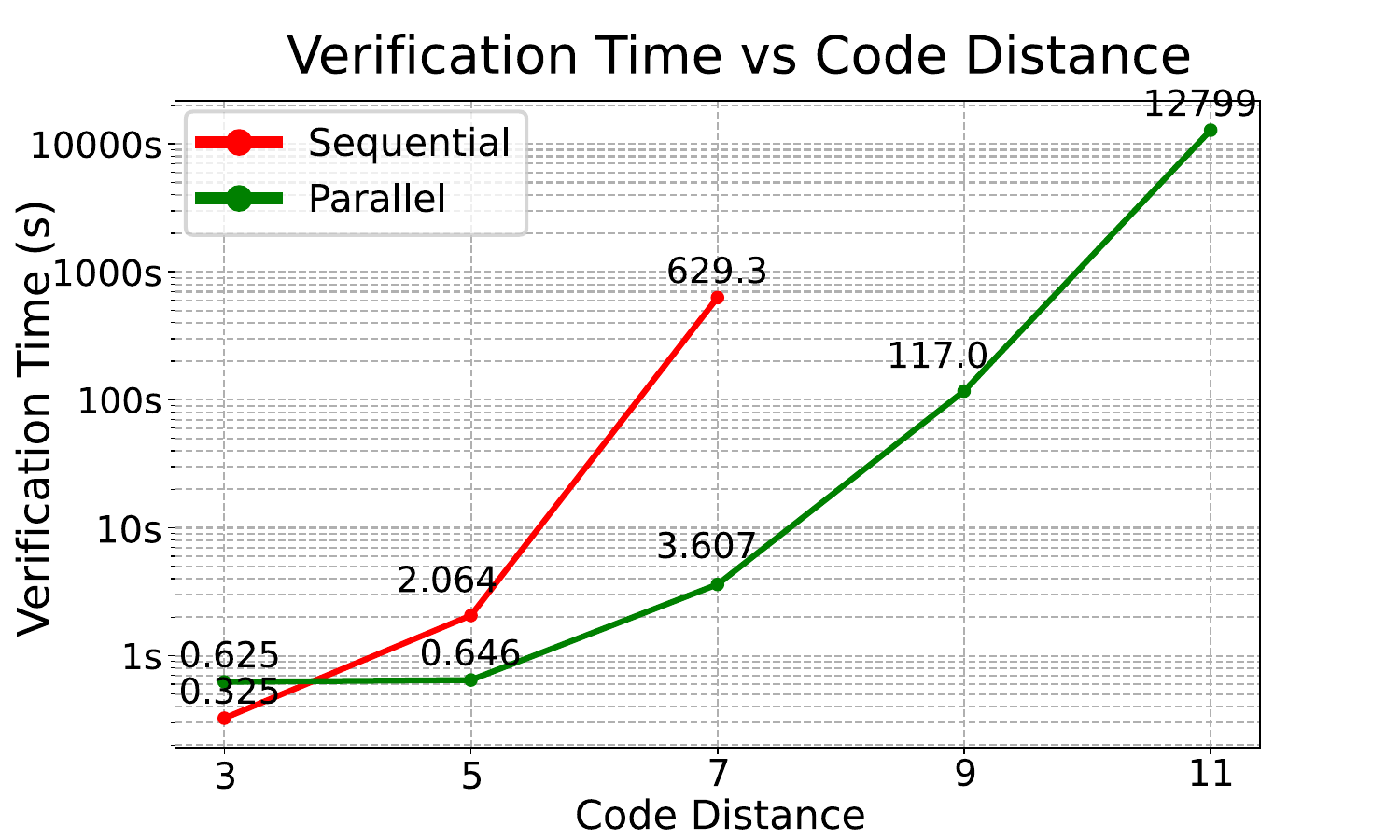}
\caption{Time consumed when verifying surface code in sequential/parallel.}
\label{fig:verify1}
\end{minipage}  
\quad
\begin{minipage}{0.53\textwidth}
    \centering
    \includegraphics[width = \linewidth]{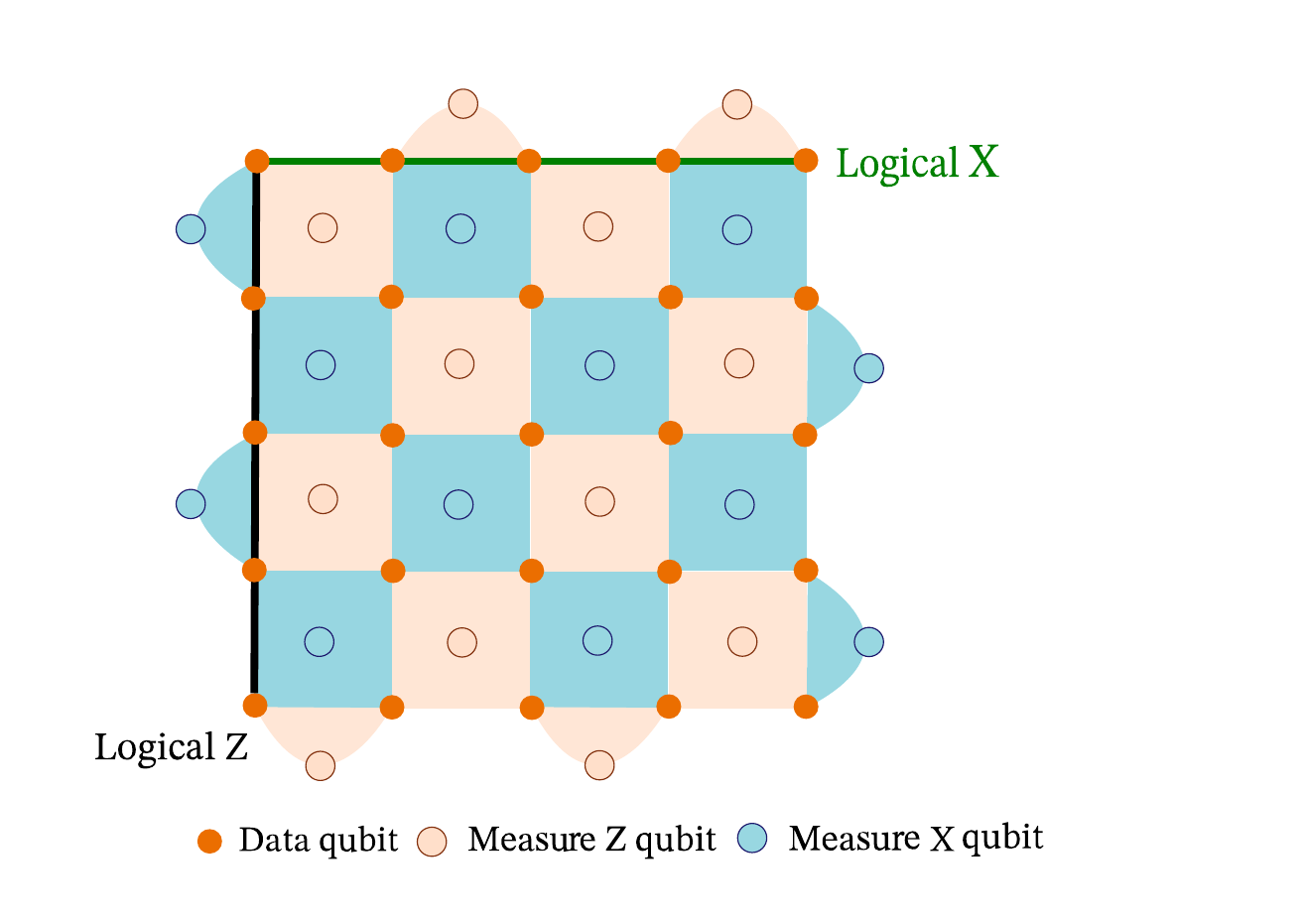}
    \caption{Scheme of a rotated surface code with $d = 5$. Each coloured tile associated with the measure qubit in the center is a stabilizer (Flesh: $Z$ check, Indigo: $X$ check). }
    \label{fig:surface}    
    \end{minipage}
\vspace{-0.05cm}
\end{figure}

\begin{figure}
\centering
\includegraphics[width = .89\linewidth]{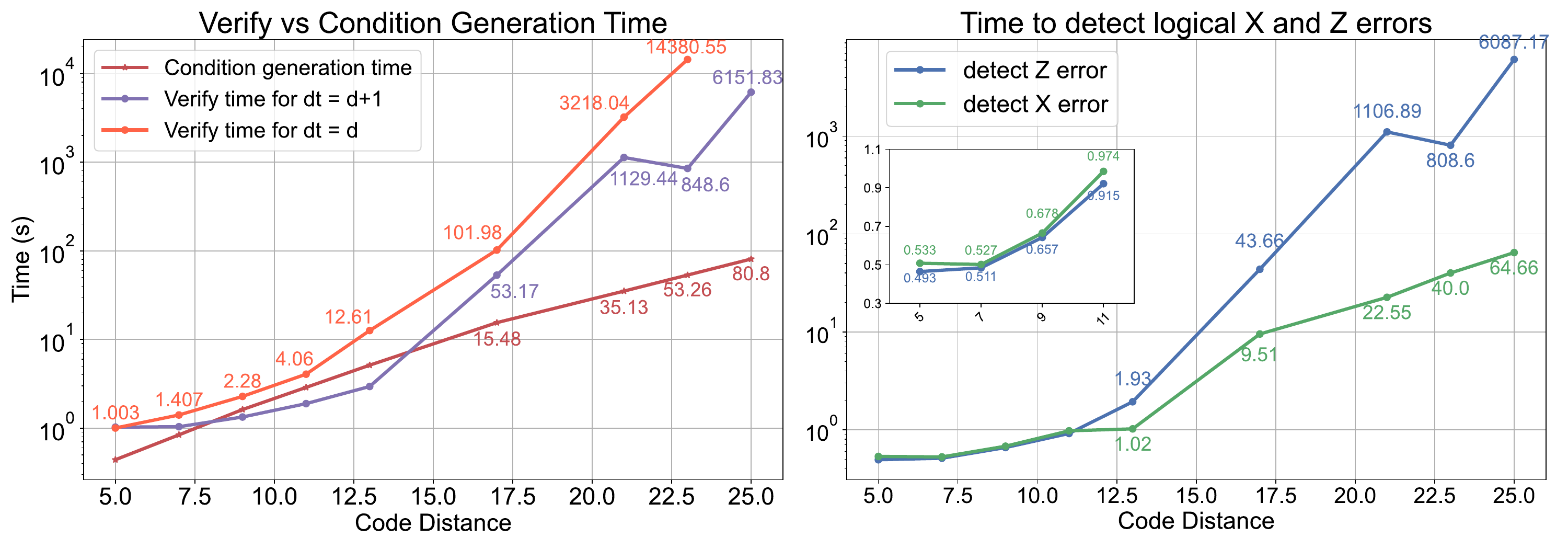}
\vspace{-0.2cm}
\caption{Time consumed when verifying precise detection properties on surface code with distance $d$.}
\vspace{-0.3cm}
\label{fig:verify2}
\end{figure}

\paragraph{\textbf{Results}}
\textit{Accurate Decoding and Correction}: Fig. \ref{fig:verify1} illustrates the total runtime required to verify the error conditions for accurate decoding and correction, employing both sequential and parallel methods. The figure indicates that while both approaches produce correct results, our parallel strategy significantly improves the efficiency of the verification process. In contrast, the sequential method exceeded the maximum runtime of 24 hours at $d = 9$; we extended the threshold for solvable instances within the time limit to $d = 11$.

\textit{Precise Detection of Errors}: For a rotated surface code with distance $d$, we first set $d_t = d$ to verify that all error patterns with Hamming weights $w < d$ can be detected by the stabilizer generators. Afterward, we set $d_t = d + 1$ to detect error patterns that are undetectable by the stabilizer generators but cause logical errors. The results show that all trials with $d_t = d$ report \textit{unsat} for Eqn. \eqref{eqn:detect}, and trials with $d_t = d + 1$ report \textit{sat} for Eqn. \eqref{eqn:detect}, providing evidence for the effectiveness of this functionality. The results indicate that, without prior knowledge of the minimum weight, this tool can identify and output the minimum weight undetectable error. Fig. \ref{fig:verify2} illustrates the relationship between the time required for verifying error conditions of precise detection of errors and the code distance.

\subsection{Verify correctness with user-provided errors}
Constrained by the exponential growth of problem size, verifying general properties limits the size of QEC codes that can be analyzed. Therefore, we allow users to autonomously impose constraints on errors and verify the correctness of the QEC code under the specified constraints. We aim for the enhanced tool, after the implementation of these constraints, to increase the size of verifiable codes. Users have the flexibility to choose the generated constraints or derive them from experimental data, as long as they can be encoded into logical formulas supported by SMT solvers. The additional constraints will also help prune the solution space by eliminating infeasible enumeration paths during parallel solving.
\paragraph{\textbf{Results}}
We briefly analyze the experimental data~\cite{Acharya2023suppress, acharya2024quantumerrorcorrectionsurface} and observe that the error detection probabilities of stabilizer generators tend to be uniformly distributed. Moreover, among the physical qubits in the code, there are always several qubits that exhibit higher intrinsic single-qubit gate error rates. Based on these observations, we primarily consider two types of constraints and evaluate their effects in our experiment. For a rotated surface code with distance $d$, the explicit constraints are as follows: 
\begin{itemize}
    \item Locality: Errors occur within a set containing $\frac{d^2 -1}{2}$ randomly chosen qubits. The other qubits are set to be error-free. 
    \item Discreteness: Uniformly divide the total $d^2$ qubits into $d$ segments, within each segment of $d$ qubits there exists no more than one error. 
\end{itemize}
The other experimental settings are the same as those in the first experiment. 

Fig. \ref{fig:test} illustrates the experimental results of verification with user-provided constraints. We separately assessed the results and the time consumed for verification with the locality constraint, the discreteness constraint, and both combined. We will take the average time for five runs for locality constraints since the locations of errors are randomly chosen. Obviously both constraints contribute to the improvement of efficiency, yet yield limited improvements if only one of them is imposed; When the constraints are imposed simultaneously, we can verify the $d = 19$ surface code which has $361$ qubits within $\sim 100$ minutes. 

\paragraph{\textbf{Comparison with \textsc{Stim}~\cite{gidney2021stim}}}
\textsc{Stim} is currently the most widely used and state-of-the-art stabilizer circuit simulator that provides fast performance in sampling and testing large-scale QEC codes.
However, simply using \textsc{Stim} in sampling or testing does not provide a complete check for QEC codes, as it will require a large number of samples. For example, we can verify a $d = 19$ surface code with $361$ qubits in the presence of both constraints, which require testing on $\sum_{i=0}^{18}\binom{18}{i}(18)^i = 19^{18} \approx 2^{76}$ samples that are beyond the testing scope.
\begin{figure}
\begin{minipage}{0.47\linewidth}
\centering
\includegraphics[width = 0.95\linewidth]{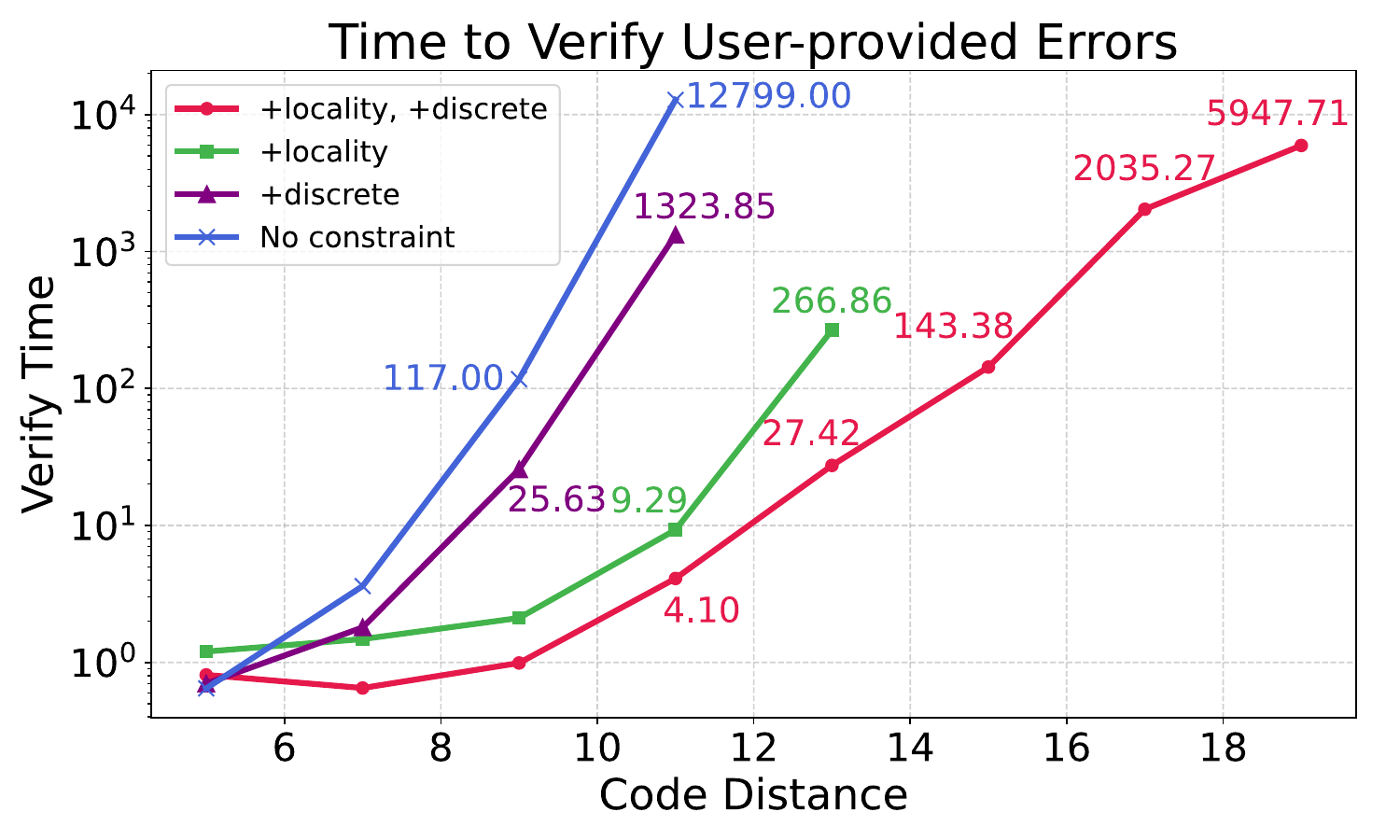}
\caption{Time consumed to verify the correctness of surface code with distances ranging from $5$ to $19$.}
\label{fig:test}
\end{minipage}
\quad
\begin{minipage}{0.47\linewidth}
\centering
\includegraphics[width = \linewidth]{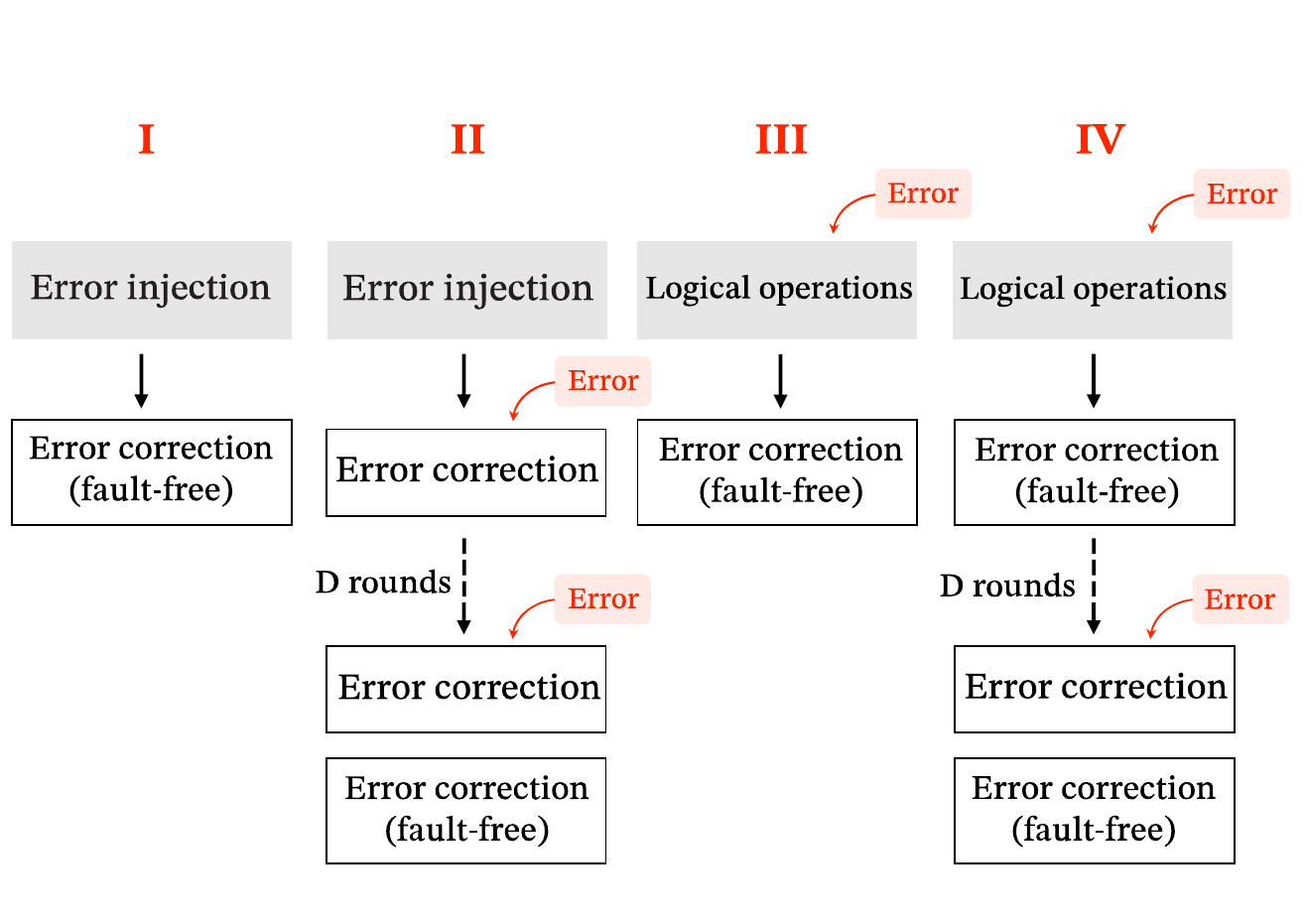}
\caption{Realistic fault-tolerant scenarios that are supported for verification.}
\label{fig:sce}
\end{minipage}
\end{figure}
\subsection{Towards fault-tolerant implementation of operations in quantum hardware}
We are interested in whether our tool has the capability to verify the correctness of fault-tolerant implementations for certain logical operations or measurements. In Fig. \ref{fig:sce} we conclude the realistic fault-tolerant computation scenarios our tools support. In particular, we write down the programs of two examples encoded by Steane code and verify the correctness formulas in our tool. The examples are stated as follows:
\begin{enumerate}
  \item A fault-tolerant logical GHZ state preparation.
    \item An error from the previous cycle remains uncorrected and got propagated through a logical $CNOT$ gate. 
\end{enumerate} 
We provide the programs used in the experiment in Fig. \ref{fig:GHZ} and Fig. \ref{fig:cnot}. The program \textbf{Steane}$(E)_i$ denotes an error correction process over $i^{\rm{th}}$ logical qubit encoded using the Steane code. 
\begin{figure}
\begin{minipage}{0.45\linewidth}
\centering
\vspace{-0.1cm}
\begin{align*}
&\qqfor{8 \cdots 14}{\qut{H}{q_i}} \ \fatsemi  \\
& \qqfor{1\cdots 3}{\textbf{Steane}(E)_{i}} \ \fatsemi \\
& \qqfor{8\cdots 14}{\qut{CNOT}{q_{i}, q_{i-7}}}\ \fatsemi \\
& \qqfor{1\cdots 7}{\qut{CNOT}{q_{i}, q_{i+7}}} \ \fatsemi \\
& \qqfor{1\cdots 3}{\textbf{Steane}(E)_{i}} \\[-0.83cm]
\end{align*}
\caption{QEC for logical GHZ state preparation.}
\label{fig:GHZ}
\end{minipage}
\qquad
\begin{minipage}{0.45\linewidth}
\centering
\begin{align*}
&\\[-0.3cm]
& \qqfor{1\cdots 7}{\qut{U}{[ep_{(i)}]q_i}} \ \fatsemi \\
& \qqfor{1\cdots 7}{\qut{CNOT}{q_{i}, q_{i + 7}}} \ \fatsemi \\
& \qqfor{1\cdots 2}{\textbf{Steane}(E)_{i}}\\
&\\[-1.03cm]
\end{align*}
\caption{QEC for logical CNOT gate with propagated errors.}
\label{fig:cnot}
\end{minipage}

\end{figure}

\subsection{A benchmark for qubit stabilizer codes}
We further provide a benchmark of $14$ qubit stabilizer codes selected from the broader quantum error correction code family, as illustrated in Table \ref{tab:benchmark}. We require the selected codes to be qubit-based and have a well-formed parity-check matrix. For codes that lack an explicit parity-check matrix, we construct the stabilizer generators and logical operators based on their mathematical construction and verify the correctness of the implementations. For codes with odd distances, we verify the correctness of their program implementations in the context of accurate decoding and correction. However, some codes have even code distances, including examples such as the 3D $[[8,3,2]]$ color code~\cite{Kubica_2015} and the Campbell-Howard code~\cite{campbell2017unified}, which are designed to implement non-Clifford gates like the $T$-gate or Toffoli gate with low gate counts. These codes have a distance of $2$, allowing error correction solely through post-selection rather than decoding. In such cases, the correctness of the program implementations is ensured by verifying that the code can successfully detect any single-qubit Pauli error. We list these error-detection codes at the end of Table \ref{tab:benchmark}.

\begin{table}
\caption{A benchmark of qubit stabilizer codes with logical-free scenario ($EMC$) considered in Table \ref{tab:comp}. We report their parameters $[[n, k, d]]$ and the properties we verified with the time spent. Parameters with variables indicate that this code has a scalable construction. If the exact $d$ is unknown, we provide an estimation given by our tool in the bracket.}
\label{tab:benchmark}
\vspace{-0.3cm}
\renewcommand{\arraystretch}{1}
{\begin{tabular}{|c|c|c|}
\hline
\multicolumn{3}{|c|}{\textbf{Target: Accurate Correction}} \\ \hline
Code Name & Parameters & Verify time(s) \\ 
\hline 
Steane code~\cite{steane1996multiple} &$[[7,1,3]]$ & 0.095 \\
Surface code~\cite{dennis2002topological} ($d = 11$) &  $[[d^2, 1, d]]$  & 12799 \\ 
Six-qubit code~\cite{calderbank1997quantum} &$[[6,1,3]]$  & 0.252 \\
Quantum dodecacode~\cite{calderbank1997quantum} & $[[11,1,5]]$  & 0.587 \\
Reed-Muller code~\cite{steane1996quantumreed} ($r = 8$) & $[[2^r - 1, 1, 3]]$  & 1868.56 \\
XZZX surface code~\cite{BonillaAtaides2021} ($d_x = 9, d_z = 11$) & $[[d_x\times d_z, 1, \min(d_x, d_z)]]$  & 1067.16 \\
Gottesman code ~\cite{gottesman1997stabilizer} ($r = 8$) & $[[2^r, 2^r - r -2, 3]]$  & 587.00\\
Honeycomb code~\cite{landahl2011fault} ($d = 5$) & $[[19, 1, 5]]$ & 1.55 \\
\hline
\multicolumn{3}{|c|}{\textbf{Target: Detection}} \\ \hline
Tanner Code I~\cite{leverrier2022quantumtannercodes} & $[[343, 31, d \geq 4]]$ & $7086.36$\\
Tanner Code II~\cite{leverrier2022quantumtannercodes} &$[[125, 53, 4]]$ & $1667.81$ \\
\hline
Hypergraph Product~\cite{tillich2014ldpc, Kovalev2012hgp, breuckmann2021quantum} & $[[98,18,4]]$ & 289.37   \\ \hline
\multicolumn{2}{|l|}{\textit{Error-Detection codes}}  &  \\ \hline
3D basic color code~\cite{Kubica_2015} ($d_z = 2$) & $[[8,3,2]]$ & 2.88  \\
Triorthogonal code~\cite{bravyi2012magic} $(k = 64)$ & $[[3k+8,k,d_x = 6, d_z = 2]] $ & $144.94$ \\
Carbon code~\cite{grassl2013leveraging} & $[[12,2,4]]$  & 4.80\\
Campbell-Howard code~\cite{campbell2017unified} ($k = 2$) & $[[6k +2, 3k, 2]]$ & 3.05 \\
\hline
\end{tabular}}
\vspace{-0.1cm}
\end{table}

\newcommand{\qbricks}{{\textsc{Qbricks}}}
\newcommand{\qwire}{{$\mathcal{Q}$\textsc{wire}}}
\newcommand{\reqwire}{{$Re\mathcal{Q}$\textsc{wire}}}
\newcommand{\sqir}{s\textsc{qir}\xspace}
\section{Related Work}\label{sec:related_work}
In addition to the works compared in Section~\ref{introduction}, we briefly outline verification techniques for quantum programs and other works that may be used to check QEC programs.

\paragraph{Formal verification with program logic}
Program logic, as a well-explored formal verification technique, plays a crucial role in the verification of quantum programs. Over the past decades, numerous studies have focused on developing Hoare-like logic frameworks for quantum programs~\cite{baltag2004logic,brunet2004dynamic,chadha2006reasoning,feng2007proof,kakutani2009logic}.
\textit{Assertion Logic}. \cite{wu2021qecv, Rand2021gottesman,Rand2021StaticAO} began utilizing stabilizers as atomic propositions. \cite{unruh2019quantum} proposed a hybrid quantum logic in which classical variables are embedded as special quantum variables. Although slightly different, this approach is essentially isomorphic to our interpretation of logical connectives.
\textit{Program Logic}. 
Several works have established sound and relatively complete (hybrid) quantum Hoare logics, both satisfaction-based~\cite{ying2012floyd, feng2021quantum} and projection-based~\cite{zhou2019applied}. However, these works did not introduce (countable) assertion syntax, meaning their completeness proofs do not account for the expressiveness of the weakest (liberal) preconditions.
\cite{wu2021qecv,wu2024towards,sundaram2022hoare} focus on reasoning about stabilizers and QEC code, with our substitution rules for unitary statements drawing inspiration from their work.
\textit{Program logic in the verification of QEC codes and fault-tolerant computing}. Quantum relational logic~\cite{unruh2019quantumr,LU21,BHY19} is designed for reasoning about relationships, making it well-suited for verifying functional correctness by reasoning equivalence between ideal programs and programs with errors. 
Quantum separation logic~\cite{LLSS22,ZBH21,qafny,qstar}, through the application of separating conjunctions, enables local and modular reasoning about large-scale programs, which is highly beneficial for verifying large-scale fault-tolerant computing. Abstract interpretation~\cite{yu2021quantum} uses a set of local projections to characterize properties of global states, thereby breaking through the exponential barrier. It is worth investigating whether local projections remain effective for QEC codes.

\vspace{-0.5em}

\paragraph{Symbolic techniques for quantum computation}
General quantum program testing and debugging methods face the challenge of excessive test cases when dealing with QEC programs, which makes them inefficient.
Symbolic techniques have been introduced into quantum computing to address this issue~\cite{carette2023symbolic,bauer2023symqv,10.1145/3519939.3523431,chen2023automata,10.1145/3445814.3446750,Fang2024symbolic,fang2024symphase}.
Some of these works aim to speed up the simulation process without providing complete verification of quantum programs, while others are designed for quantum circuits and do not support the control flows required in QEC programs.
The only approach capable of handling large-scale QEC programs is the recent work that proposed symbolic stabilizers~\cite{Fang2024symbolic}.
However, this framework is primarily designed to detect bugs in the error correction process that do not involve logical operations and do not support non-Clifford gates.

\vspace{-0.5em}

\paragraph{Mechanized approach for quantum programming}
The mechanized approach offers significant advantages in terms of reliability and automation, leading to the development of several quantum program verification tools in recent years (see recent reviews~\cite{CVLreview, LSZreview}). Our focus is primarily on tools that are suitable for writing and reasoning about quantum error correction (QEC) code at the circuit level.
\textit{Matrix-based approaches}. \qwire~\cite{paykin2017qwire,RPZ17} and \sqir~\cite{hietala2021verified} formalize circuit-like programming languages and low-level languages for intermediate representation, utilizing a density matrix representation of quantum states. These approaches have been extended to develop verified compilers~\cite{rand2019reqwire} and optimizers~\cite{hietala2021verified}.
\textit{Graphical-based approaches}. \cite{VyZX22,VyZX23,shah2024vicarvisualizingcategoriesautomated}, provide a certified formalization of the ZX-calculus~\cite{coecke2011interacting,PyZX}, which is effective for verifying quantum circuits through a flexible graphical structure.
\textit{Automated verification}. 
\qbricks~\cite{chareton2020deductive} offers a highly automated verification framework based on the Why3~\cite{bobot2011why3} prover for circuit-building quantum programs, employing path-sum representations of quantum states~\cite{amy2018towards}. 
\textit{Theory formalization}. 
Ongoing libraries are dedicated to the formalization of quantum computation theories, such as QuantumLib~\cite{zweiflerquantumlib}, Isabelle Marries Dirac (IMD)~\cite{BLH21, bordg2020isabelle}, and CoqQ~\cite{zhou2023coqq}. QuantumLib is built upon the Coq proof assistant and utilizes the Coq standard library as its mathematical foundation. IMD is implemented in Isabelle/HOL, focusing on quantum information theory and quantum algorithms. CoqQ is written in Coq and provides comprehensive mathematical theories for quantum computing based on the Mathcomp library~\cite{Mathcomp,mathcomp-analysis}. Among these, CoqQ has already formalized extensive theories of subspaces, making it the most suitable choice for our formalization of program logic.

\vspace{-0.5em}

\paragraph{Functionalities of verification tools for QEC programs}
Besides the comparison of theoretical work on program logic and other verification methods, we also compare the functionalities of our tool \veriq with those of other verification tools for QEC programs. We summarize the functionalities of the tools in Table~\ref{tab:comp}. VERITA~\cite{wu2021qecv,wu2024towards} adopts a logic-based approach to verify the implementation of logical operations with fixed errors. \texttt{QuantumSE}~\cite{Fang2024symbolic} is tailored for efficiently reporting bugs in QEC programs and shows potential in handling logical Clifford operations. Stim~\cite{gidney2021stim} employs a simulation-based approach, offering robust performance across diverse fault-tolerant scenarios but limited to fixed errors. Our tool \veriq{} is designed for both general verification and partial verification under user-provided constraints, supporting all aforementioned scenarios.
\begin{table}
\caption{Comparison of scenarios and functionalities between \veriq and other tools. For scenarios, we denote $\bar{L}$ for logical gate implementation, $E$ for error injection, $M$ for measurement (error detection),
$C$ ($C_E$) for error correction (with error injection).
We further identify three functionalities, $\mathbf{C}$ for general verification of correctness, $\mathbf{R}$ for reporting bugs, and $\mathbf{F}$ for fixed errors, that evaluated by $\blacktriangle$ if implemented,
    $\circ$ if potentially supported but not yet implemented
    and $-$ if cannot handle or beyond design. n/a indicates that $\mathbf{F}$ is unavailable in the error-free scenario.}
    \vspace{-0.2cm}
    \centering

    \begin{tabular}{|c|*{3}{>{\centering\arraybackslash}p{0.36cm}|}*{3}{>{\centering\arraybackslash}p{0.36cm}|}*{3}{>{\centering\arraybackslash}p{0.36cm}|}*{3}{>{\centering\arraybackslash}p{0.36cm}|}}
        \hline
        \multirow{2}{*}{\diagbox{Scenarios}{Tools}} & \multicolumn{3}{c|}{\multirow{2}{*}{\veriq{}}} & \multicolumn{3}{c|}{\textsc{VERITA}} & \multicolumn{3}{c|}{\textsf{QuantumSE}} & \multicolumn{3}{c|}{\textsc{Stim}} \\
        & \multicolumn{3}{c|}{} & \multicolumn{3}{c|}{\cite{wu2024towards,wu2021qecv}} & \multicolumn{3}{c|}{\cite{Fang2024symbolic}} & \multicolumn{3}{c|}{\cite{gidney2021stim}} \\
        \hline
       Functionality&$\mathbf{C}$&\mbox{\textbf{R}}&\textbf{F}& \textbf{C}&\textbf{R}&\textbf{F}&\textbf{C}&\textbf{R}&\textbf{F}&\textbf{C}&\textbf{R}&\textbf{F}\\
        \hline
        \cline{1-13}
        \cline{1-13}
        error-free ($\bar{L}$) & $\blacktriangle$ &$\circ$ & n/a &$\blacktriangle$&$\circ$ &n/a& $\circ$& $\circ$  &n/a&$\circ$ &$\circ$&n/a \\
        \hline
        logical-free ($EMC$) & $\blacktriangle$ & $\circ$ & $\circ$ & $-$ & $-$ & $\blacktriangle$ & $\blacktriangle$ & $\blacktriangle$ & $\circ$ & $-$ & $-$ & $\blacktriangle$ \\
        \hline 
        error in correction step ($\bar{L} MC_E$) & $\blacktriangle$&  $\circ$ & $\circ$ & $-$ & $-$ & $\circ$ & $\blacktriangle$ & $\blacktriangle$ & $\circ$ & $-$ & $-$ & $\blacktriangle$ \\
        \hline
        one cycle ($E\bar{L} EMC$) & $\blacktriangle$ &$\circ$ & $\circ$ & $-$& $-$ & $\blacktriangle$ & $\blacktriangle$ &$\blacktriangle$ &$\circ$  & $-$ &$-$& $\blacktriangle$ \\
        \hline
        multi cycles ($E\bar{L}EMCE\bar{L}EMC \cdots$) & $\blacktriangle$& $\circ$ &$\circ$ & $-$ &$-$ &$\blacktriangle$ & $\circ$ & $\circ$ & $\circ$ & $-$ & $-$ & $\blacktriangle$ \\
        \hline
    \end{tabular}
    \label{tab:comp}
\end{table}

\section{Discussion and Future Works}

In this paper, we propose an efficient verification framework for QEC programs, within which we define the assertion logic along with program logic and establish a sound proof system. We further develop an efficient method to handle verification conditions of QEC programs. We implement our QEC verifiers at two levels: a verified QEC verifier and a Python-based automated QEC verifier.

Our work still has some limitations. First of all, the gate set we adopt in the programming language is restricted, and the current projection-based logic is unable to reason about probabilities. Last but not least, while our proof system is sound, its completeness- especially for programs with loops- remains an open question.

Given the existing limitations, some potential directions for future advancements include:
\begin{enumerate}
\item \textit{Addressing the completeness issue of the proof system.} We are able to prove the (relative) completeness of our proof system for finite QEC programs without infinite loops. However, it is still open whether the proof system is complete for programs with while-loops. This issue is indeed related to the next one.

\item \textit{Extending the gate set to enhance the expressivity of program logic}. The Clifford + $T$ gate set we use in the current program logic is universal but still restricted in practical applications. It is desirable to extend the syntax of factors and assertions for the gate sets beyond Clifford + $T$.

\item \textit{Generalizing the logic to satisfaction-based approach.} Since any Hermitian operator can be written as linear combinations of Pauli expressions, our logic has the potential to incorporate the so-called satisfaction-based approach with Hermitian operators as quantum predicates, which helps to reason about the success probabilities of quantum QEC programs.
\item \textit{Exploring approaches to implementing an automatic verified verifier.} The last topic is to explore tools like $F^{*}$~\cite{swamy2016dependent, martínez2019metafproofautomationsmt}, a proof-oriented programming language based on SMT, for incorporating the formally verified verifier and the automatic verifier described in this paper into a single unified solution.
\end{enumerate}

\section*{Acknowledgement}
We thank Bonan Su for kind discussions regarding on crafting the introduction section and Huiping Lin for for the revisions made to the introduction of stabilizer codes. In addition, we thank anonymous referees for helpful comments and suggestions.
This research was supported by the National Key R\&D Program of China under Grant No. 2023YFA1009403.

\vspace{-0.2cm}

\section*{Data Availability Statement}
The code for of this work (both the Coq formalization and the automatic verifier \veriq) is available at \url{https://github.com/Chesterhuang1999/Veri-qec}, or at \url{ https://doi.org/10.5281/zenodo.15248774} (evaluated artifact~\cite{huang_2025artifact}). The appendices are provided as the supplementary material, or see our extended version~\cite{Huang2025efficientextend}.

\vspace{-0.1cm}

\bibliographystyle{ACM-Reference-Format}
\bibliography{refs}
\clearpage
\appendix
\section{Supplementary Materials for Section \ref{assert-logic} and Section 4} \label{supp-assert}
Here we provide technical details for Section 3 regarding the assertion logic in Section \ref{assert-logic}. All lemmas and theorems are proved in our Coq implementation based on CoqQ~\cite{zhou2023coqq}.

\subsection{A Syntax of Basic Expression}\label{comp-expr}

We first claim the expressivity of $SExp$ and $PExp$ discussed in the main context.
\begin{proposition}[Expressivity of $SExp$ and $PExp$]
    Any constant $s\in \mathbb{Z}[\frac{1}{\sqrt{2}}]$ is expressible in $SExp$. Any constant $W$ belonging to the Pauli group over qubits $1,\cdots, n$ is expressible in $PExp$.
\end{proposition}

We further specify the boolean expressions $BExp$ and integer expressions $IExp$ for \veriq\  as: 
\begin{align*}
IExp:\quad a \Coloneqq\ &n\in\mathbb{N} \mid x \mid - a \mid a_1+a_2 \mid a_1\times a_2\\ 
BExp:\quad  b \Coloneqq\ & \mathbf{true} \mid \mathbf{false} \mid x \mid a_1 == a_2 \mid a_1 \leq a_2\\ 
& \mid \neg b \mid b_1 \wedge b_2 \mid b_1 \vee b_2 \mid b_1 \rightarrow b_2 .
\end{align*}
Here, $n$ are constant natural numbers, $x$ appears in $IExp$ and $BExp$ are program variables of type integer and bool, respectively.
There exists type coercion between $BExp$ and $IExp$: boolean value $\mathbf{true}$ and $\mathbf{false}$ are identified with 1 and 0, respectively. Their semantics $\sem{a}_m$ and $\sem{b}_m$ are defined conventionally as a mapping from classical state $m\in\mathtt{CMem}$ to integers and bools:
\begin{align*}
    &\sem{n}_m \triangleq n, \quad 
     \sem{x}_m \triangleq m(x), \quad 
     \sem{-a}_m \triangleq -\sem{a}_m, \\
    &\sem{a_1+a_2}_m \triangleq \sem{a_1}_m+\sem{a_2}_m, \quad
     \sem{a_1\times a_2}_m \triangleq \sem{a_1}_m\times\sem{a_2}_m, \\
    &\sem{\mathbf{true}}_m \triangleq \mathbf{true}, \quad
     \sem{\mathbf{false}}_m \triangleq \mathbf{false}, \quad 
     \sem{x}_m \triangleq m(x), \quad \\
    &\sem{a_1 == a_2}_m \triangleq \sem{a_1}_m == \sem{a_2}_m, \quad
     \sem{a_1 \leq a_2}_m \triangleq \sem{a_1}_m \leq \sem{a_2}_m, \\
    &\sem{\neg b}_m \triangleq \neg \sem{b}_m, \quad
     \sem{b_1 \wedge b_2}_m \triangleq \sem{b_1}_m \wedge \sem{b_2}_m, \\
    &\sem{b_1 \vee b_2}_m \triangleq \sem{b_1}_m \vee \sem{b_2}_m, \quad
     \sem{b_1 \rightarrow b_2}_m \triangleq \sem{b_1}_m \rightarrow \sem{b_2}_m.
\end{align*}

\subsection{The Pauli Expression is Closed under Basic Unitary Transformation}

To provide proof rules for the unitary transformation of single-qubit gates $U_1\in\{X,Y,Z,H,S,T\}$ and two-qubit gates $U_2\in\{CNOT, CZ, iSWAP\}$ for the program logic, we need first examine the properties that, for any $P\in PExp$, is $U_{1i}^\dag\sem{P}_m U_{1i}$ and $U_{2ij}^\dag\sem{P}_m U_{2ij}$ expressible in $PExp$? Here, we give an affirmative result stated below:

\begin{theorem}[Theorem 3.1]
For any Pauli expression $P$ defined in Eqn. \eqref{pauli-exp} and single-qubit gate $U_1$ acts on $q_i$ or two-qubit gate $U_2$ acts on $q_iq_j$, their exists another Pauli expression $Q\in PExp$, such that for all $m\in\mathtt{CMem}$:
$$
\sem{Q}_m = U_{1i}^\dag\sem{P}_m U_{1i},\quad \mbox{or}, \quad \sem{Q} = U_{2ij}^\dag\sem{P}_m U_{2ij}.
$$
\end{theorem}

\begin{proof}
We prove it by induction on the structure of $PExp$. The proofs of all gates are similar, we here only present the case for $T$ gate and $CNOT$ gate.

\noindent$\bullet$\ ($T$ gate). Define the substitution of any $P\in PExp$ as
$$P'\triangleq P\big[\sq(X_i - Y_i) /X_i, \sq(X_i + Y_i)/Y_i\big],$$
where $i$ is the qubit $q_i$ the $T$ gate acts on, and $P[e_1/x_1,e_2/x_2,\cdots]$ are simultaneous substitution of constant constructor $x\in \{X_r,Y_r,Z_r\}$ with expression $e$ in $P$. We then show that $P'$ is the desired $Q$.

\textbf{Base case.} For elementary expression $P\equiv p_r$, if $r \neq i$, then:
$$T^\dag_{q_i}\sem{p_r}_mT_{q_i} = I_1\otimes\cdots\otimes T^\dag_{i} I_i T_i\otimes\cdots\otimes p_r\otimes\cdots I_n = I_1\otimes\cdots\otimes I_i \otimes\cdots\otimes p_r\otimes\cdots I_n = \sem{p_r}_m = \sem{p_r'}_m,$$
i.e., we do not need to change $p_r$ in the case of $r \neq i$.
On the other hand, note that:
$$ T^\dag X T = \sq(X-Y),\quad T^\dag Y T = \sq(X+Y),\quad T^\dag Z T = Z,$$
so we obtain:
\begin{align*}
&T^\dag_{q_i}\sem{X_i}_mT_{q_i} = I_1\otimes\cdots\otimes \sq(X_i-Y_i)\otimes\cdots \otimes I_n = \sem{\sq(X_i - Y_i)}_m = \sem{X_i'}_m  \\
&T^\dag_{q_i}\sem{Y_i}_mT_{q_i} = \sem{\sq(X_i + Y_i)}_m = \sem{Y_i'}_m, \qquad
T^\dag_{q_i}\sem{Z_i}_mT_{q_i} = \sem{Z_i}_m = \sem{Z_i'}_m.
\end{align*}

\textbf{Induction step.} $P\equiv sP$. Note that 
$$T^\dag_{q_i}\sem{sP}_mT_{q_i} = T^\dag_{q_i}\sem{s}_m\sem{P}_mT_{q_i} = \sem{s}_m(T^\dag_{q_i}\sem{P}_mT_{q_i}) = \sem{s}_m\sem{P'}_m = \sem{sP'}_m = \sem{(sP)'}_m.$$

$P\equiv P_1 + P_2$. Observe that 
\begin{align*}
    T^\dag_{q_i}\sem{P_1 + P_2}_mT_{q_i} &= T^\dag_{q_i}(\sem{P_1}_m + \sem{P_1}_m)T_{q_i} = 
T^\dag_{q_i}\sem{P_1}_mT_{q_i} + T^\dag_{q_i}\sem{P_1}_mT_{q_i} \\
&= \sem{P_1'}_m + \sem{P_2'}_m = \sem{P_1'+P_2'}_m = \sem{(P_1+P_2)'}_m.
\end{align*}

$P\equiv P_1P_2$. By noticing that $T^\dag T = I$, we have:
\begin{align*}
    T^\dag_{q_i}\sem{P_1P_2}_mT_{q_i} &= T^\dag_{q_i}(\sem{P_1}_m\sem{P_2}_m)T_{q_i} = 
(T^\dag_{q_i}\sem{P_1}_mT_{q_i})(T^\dag_{q_i}\sem{P_2}_mT_{q_i}) \\
&= \sem{P_1'}_m\sem{P_2'}_m = \sem{P_1'P_2'}_m = \sem{(P_1P_2)'}_m.
\end{align*}

\noindent$\bullet$\ ($CNOT$ gate). Define the substitution of any $P\in PExp$ as
$$P'\triangleq P[X_iX_j/X_i, Y_iX_j/Y_i, Z_iY_j/Y_j, Z_iZ_j/Z_j],$$
and $P'$ is the desired $Q$.
The induction step is the same as of $T$. For the base case, we shall analyze the case that $r = i$ or $r = j$ or $r \neq i, j$. First, we observe the following facts:
\begin{align*}
    &CNOT_{ij}(X_i\otimes I_j)CNOT_{ij} = X_i\otimes X_j,\quad
     CNOT_{ij}(I_i\otimes X_j)CNOT_{ij} = I_i\otimes X_j\\
    &CNOT_{ij}(Y_i\otimes I_j)CNOT_{ij} = Y_i\otimes X_j,\quad
     CNOT_{ij}(I_i\otimes Y_j)CNOT_{ij} = Z_i\otimes Y_j\\
    &CNOT_{ij}(Z_i\otimes I_j)CNOT_{ij} = Z_i\otimes I_j,\quad
     CNOT_{ij}(I_i\otimes Z_j)CNOT_{ij} = Z_i\otimes Z_j.
\end{align*}

For $r \neq i, j$, $CNOT_{ij}^\dag\sem{p_r}_m CNOT_{ij} = \sem{p_r}_m = \sem{p_r'}_m$.

If $r = i$, then for example $X_i$, we calculate :
\begin{align*}
    CNOT_{ij}\sem{X_i}_mCNOT_{ij} &= \bigotimes_{k\neq i, j} I_k \otimes (CNOT_{ij}(X_i\otimes I_j)CNOT_{ij}) = \bigotimes_{k\neq i, j} I_k \otimes X_i\otimes X_j \\
    &= \big(\bigotimes_{k\neq i} I_k \otimes X_i\big)\big(\bigotimes_{k\neq j} I_k \otimes X_j\big) =
    \sem{X_iX_j}_m = \sem{X_i'}_m.
\end{align*}
The rest cases $Y_i,Z_i$ and $X_j, Y_j, Z_j$ are similar.
\end{proof}

\newcommand{\si}{\rightsquigarrow}
\renewcommand{\sc}{\doublecap}
\subsection{A Brief Review of Hilbert Subspace} \label{review-subspace}

We first briefly review the basic operations regarding subspaces of Hilbert space $\cH$.
Since we focus on the finite-dimensional case, any subspace of $\cH$ is always closed.
\begin{itemize}
    \item (span) Given a set of states $S\subseteq \cH$, its span $\Span\{S\} \in \cS(\cH)$ is defined by
    $$\Span\{S\} = \Big\{\sum_{i\in I}\lambda_i|\phi_i\>: I\mbox{ is a finite index set, }\lambda_i\in \mathbb{C}\mbox{, and }|\phi_i\>\in S\Big\}.$$
    \item (kernel) Given a linear operator $A$ on $\cH$, its kernel $\ker(A) \in \cS(\cH)$ is defined by
    $$\ker(A) = \{|\psi\>\in\cH : A|\psi\> = 0\}.$$
    \item (+1-eigenspace) Given a linear operator $A$ on $\cH$, its +1-eigenspace $E_1(A) \in \cS(\cH)$ is defined by
    $$E_1(A) = \{|\psi\>\in\cH : A|\psi\> = |\psi\>\}.$$
    \item (complement, or orthocomplement) For a given subspace $S\in \cS(\cH)$, its orthocomplement $S^\bot \in \cS(\cH)$ is defined by
    $$S^\bot = \{|\psi\> : \forall\ |\phi\>\in S,\ |\psi\> \perp |\phi\> \}.$$
    Orthocomplement is involutive, i.e., $S^{\bot\bot} = S$.
    \item (support) Given a linear operator $A$ on $\cH$, its support $\mathrm{supp}(A) \in \cS(\cH)$ is defined as the orthocomplement of its kernel, i.e., $\mathrm{supp}(A) = \ker(A)^\bot$. Support is idempotent, i.e., 
    $\Supp(\Supp(A)) = \Supp(A)$.
    \item (meet, or intersection, or disjunction) Given two subspaces $S, T\in \cS(\cH)$, their meet $S\wedge T \in \cS(\cH)$ is defined as the intersection:
    $$S\wedge T = S\cap T \equiv \{|\phi\> : |\phi\>\in S\mbox{ and } |\phi\>\in T\}.$$
    \item (join, or conjunction, or span of the union) Given two subspaces $S, T\in \cS(\cH)$, their join $S\vee T \in \cS(\cH)$ is defined as:
    $$S\vee T = \Span\{ S\cup T \}.$$
    It holds that: $(S\vee T)^\bot = S^\bot \wedge T^\bot$ and $(S\wedge T)^\bot = S^\bot \vee T^\bot$. Generally, there is no distributivity of $\vee$ and $\wedge$.
    \item (commute) Given two subspaces $S, T\in \cS(\cH)$, we say $S$ commutes with $T$, written $S\mathtt{C}T$, if $S = (S\wedge T)\vee (S\wedge T^\bot).$ Commutativity plays an essential role in reasoning about Hilbert space. Some properties include: 
    \begin{align*}
        &S\mathtt{C}T\mbox{ iff }T\mathtt{C}S, \quad S\mathtt{C}S, \quad S\subseteq T\mbox{ implies }S\mathtt{C}T, \quad S\mathtt{C} T\mbox{ implies }S\mathtt{C}T^\bot.
    \end{align*}
    Distributivity of meet and join holds when commutativity is assumed: if \emph{two} of $S\mathtt{C}T_1, S\mathtt{C}T_2, T_1\mathtt{C}T_2$ hold, then:
    $$S\wedge (T_1\vee T_2) = (S\wedge T_1) \vee (S\wedge T_2), \quad S\vee (T_1\wedge T_2) = (S\vee T_1) \wedge (S\vee T_2).$$
    \item (Sasaki implication) Given two subspaces $S, T\in \cS(\cH)$, the Sasaki implication $S\rightsquigarrow T \in \cS(\cH)$ is defined by 
    $$S\rightsquigarrow T = S^\bot\vee (S\wedge T).$$
    Sasaki implication is viewed as an extension of classical implication in quantum logic since it satisfies Birkhoff-von Neumann requirement: $S \rightsquigarrow T = I$ if and only if $S \subseteq T$, and the compatible import-export law: if $S$ commutes with $T$, then for any $W$, $S \wedge T \subseteq W$ if and only if $S \subseteq W \rightsquigarrow T$.
    \item (Sasaki projection) Given two subspaces $S, T\in \cS(\cH)$, the Sasaki projection $S\doublecap T \in \cS(\cH)$ is defined by 
    $$S\doublecap T = S\wedge (S^\bot\vee T).$$
    Sasaki projection is a ``dual'' of implication, i.e., $(S\sc T)^\bot = S\si T^\bot$, $(S\si T)^\bot = S\sc T^\bot$. It preserves order for the second parameter, i.e., $T_1\subseteq T_2$ implies $S\sc T_1\subseteq S\sc T_2$. $\Supp(P_S A P_S) = P_S\sc \Supp(A)$ which appears useful for reasoning about measurement~\cite{feng2023refinement}.
\end{itemize}

\subsection{A Hilbert-style Proof System of Assertion Logic}\label{proof-sys-assert}

The proof system presented in Fig. \ref{proof-sys-ass} is sound for quantum logic, and thus is also sound for our assertions, as its semantics is a point-wise lifting of quantum logic. We say two assertions $A,B$ commute, written $A\mathtt{C}B$, if for all $m$, $\sem{A}_m\mathtt{C}\sem{B}_m$.
\begin{figure}
\begin{align*}
& 1. \ \inferrule{}{\neg \neg A \vdash A}  && 2. \ \inferrule{}{A\vdash A}  && 3. \ \inferrule{}{A \vdash \top} && 4. \ \inferrule{}{\bot \vdash A}  \\
& 5. \ \inferrule{\Gamma \vdash A \quad \Gamma \vdash B}{\Gamma \vdash A\wedge B}  &&6. \ \inferrule{\Gamma \vdash A_1 \wedge A_2}{\Gamma \vdash A_i} && 7. \ \inferrule{A\vdash B}{\Gamma \wedge A \vdash B}  &&8.\ \inferrule{\Gamma \vdash A \quad \Gamma^{\prime} \vdash A }{\Gamma \vee \Gamma^{\prime} \vdash A}  \\ 
&9. \ \inferrule{\Gamma \vdash A_i}{\Gamma \vdash A_1 \vee A_2} 
&& 10. \ \inferrule{A\vdash B \Rightarrow C \quad A 
\vdash B}{A \vdash C} && 11. \ \inferrule{A\wedge B \vdash C\quad A\mathtt{C}B}{A \vdash B \Rightarrow C} && 
\end{align*}
\caption{A Hilbert-style proof system for assertion logic.}
\label{proof-sys-ass}
\end{figure}

We also provide two auxiliary laws to help simplify special Pauli expressions:
\begin{proposition}
For any $P,Q \in PExp$, the following laws are correct: 
$$\mathrm{i}) \ P \wedge Q \eqmodels P \wedge QP, \qquad \mathrm{ii}) \ P \wedge - P \eqmodels \mathbf{false}.$$
\label{ass-law-app}
\end{proposition}

\subsection{Denotational Semantics of QEC Programs}
\label{sec:app-denotational-semantics}
\citet{feng2021quantum} gives the induced denotational semantics of the classical-quantum program, the structural representation of each construct is as follows:
\begin{proposition}[c.f. \cite{feng2021quantum}]
The denotational semantics for QEC programs enjoy the following structure representation:
\begin{enumerate}
    \item $\sem{\mathbf{skip}}(m,\rho) = (m,\rho)$;
    \item $\sem{q_i \coloneq \vert 0\rangle}(m,\rho) = (m,\sum_{k = 0,1}|0\>_{q_i}\<k|\rho|k\>_{q_i}\<0|)$;
    \item $\sem{\qut{U}{q_i}}(m,\rho) = (m,U_{q_i}\rho U_{q_i}^{\dag})$;
    \item $\sem{\qut{U}{q_iq_j}}(m,\rho) = (m,U_{q_{i,j}}\rho U_{q_{i,j}}^{\dag})$;
    \item $\sem{x \coloneq e}(m,\rho) = (m[\sem{e}_m/x], \rho)$;
    \item $\sem{S_1\fatsemi S_2}(m,\rho) = \sum_{o\in\mathtt{CMem}}\sem{S_2}(o,\sem{S_1}(m,\rho)(o))$;
    \item $\sem{\qassign{\mathbf{meas}[P]}{x}}(m,\rho) = (m[0/x], \mathtt{P}_{\sem{P}_m}\rho\mathtt{P}_{\sem{P}_m}) + (m[1/x], \mathtt{P}_{\sem{P}_m^\bot}\rho \mathtt{P}_{\sem{P}_m^\bot})$;
    \item $\sem{\qqif{b}{S_1}{S_0}}(m,\rho) = \left\{ \begin{aligned}
  &\sem{S_0}(m,\rho), \ b \equiv \mathbf{false}\\
  &\sem{S_1}(m,\rho), \ b \equiv \mathbf{true}
  \end{aligned}\right.$;
    \item $\sem{\qqwhile{b}{S}}(m,\rho) = \lim_n (\sem{(\mathbf{while})^n}(m,\rho)).$
\end{enumerate}
Note that projection is Hermitian, so we omit ${}^\dag$ in (7). $(\mathbf{while})^n$ is the $n$-th syntactic approximation of $\mathbf{while}$, i.e., $(\mathbf{while})^0 = \mathbf{abort}$, and $(\mathbf{while})^{(n+1)} = \qqif{b}{S\fatsemi(\mathbf{while})^n}{\mathbf{skip}}$.
As mentioned, we do not lift the input state from singleton to the general classical-quantum state, (6) is thus slightly different from \cite{feng2021quantum}. In (9), as the sequence always converges, we simply write $\lim$ instead of the least upper bound in \cite{feng2021quantum}.
\label{prop-struct}
\end{proposition}

It is alternative to express denotational semantics as $\sem{S}' : \mathtt{CMem}\rightarrow\mathtt{CMem}\rightarrow\mathcal{QO}(\mathcal{H})$; for given input and output classical state $m_{in}$ and $m_{out}$, the evolution of quantum system is described by quantum operation $\sem{S}'_{m_{in},m_{out}}$, and $\sem{S}'_{m_{in},m_{out}}(\rho) = \sem{S}(m_{in},\rho)(m_{out})$.
Some structure representations of $\sem{S}'$ are as follows:

\begin{enumerate}
    \item $\sem{\mathbf{skip}}'_{m,m} = \cI$ and $\sem{\mathbf{skip}}'_{m,m'} = 0$ if $m \neq m'$;
    \item $\sem{q_i \coloneq |0\>}'_{m,m}(\rho) = \sum_{k = 0,1}|0\>_{q_i}\<k|\rho|k\>_{q_i}\<0|$ and $\sem{q_i \coloneq |0\>}'_{m,m'} = 0$ if $m \neq m'$;
    \item $\sem{\qut{U}{q_i}}'_{m,m}(\rho) = U_{i}\rho U_{i}^{\dag}$ and $\sem{\qut{U}{q_i}}'_{m,m'} = 0$ if $m \neq m'$;
    \item $\sem{\qut{U}{q_iq_j}}'_{m,m}(\rho) = U_{ij}\rho U_{ij}^{\dag}$ and $\sem{\qut{U}{q_iq_j}}'_{m,m'} = 0$ if $m \neq m'$;
    \item $\sem{x \coloneq e}'_{m,m[\sem{e}_m/x]} = \cI$ and $\sem{x \coloneq e}'_{m,m'} = 0$ if $m[\sem{e}_m/x] \neq m'$;
    \item $\sem{S_1\fatsemi S_2}'_{i,o} = \sum_{m\in\mathtt{CMem}}\sem{S_2}_{m,o}' \circ \sem{S_1}_{i,m}'$;
\end{enumerate}

\subsection{Weakest Liberal Precondition and Definability}
In the main text, we have already defined the satisfaction relation, entailment, as well as correctness formula for $AExp$. However, for the purpose of showing the definability of the weakest liberal precondition and weak completeness of program logic, we extended the definition to its semantics domain:

\begin{definition}[Extended satisfaction relation]
Given a classical-quantum state $\mu$ and a mapping $f_A : \mathtt{CMem}\rightarrow \mathcal{S}(\mathcal{H})$, the satisfaction relation is defined as: $\mu \models f_A$ iff for all $m \in \mathtt{CMem}$, $\mu(m)\models f_A(m)$.
When $A \in AExp$, $\mu \models A$ iff $\mu \models \sem{A}$.
\end{definition}

\begin{definition}[Extended entailment] Let $f_{A_1},f_{A_2}$ be the mappings $\mathtt{CMem}\rightarrow \mathcal{S}(\mathcal{H})$. Then: 
\begin{enumerate}
    \item $f_{A_1}$ entails $f_{A_2}$, denoted by $f_{A_1}\models f_{A_2}$, if for all classical-quantum states $\mu$, $\mu\models f_{A_1}$ implies $\mu\models f_{A_2}$.
    \item $f_{A_1}$ and $f_{A_2}$ are equivalent, denoted $f_{A_1} \eqmodels f_{A_2}$, if $f_{A_1}\models f_{A_2}$ and $f_{A_2}\models f_{A_1}$.
\end{enumerate}
Whenever $A_1, A_2\in AExp$, $A_1\models A_2$ iff $\sem{A_1}\models\sem{A_2}$, and $A_1\eqmodels A_2$ iff $\sem{A_1}\eqmodels \sem{A_2}$. 
\end{definition}

\begin{definition}[Extended correctness formula]
\label{correct-formula-app}
The correctness formula for QEC programs is defined by the Hoare triple $\{f_A\} S \{f_B\}$, where $S\in Prog$ is a quantum program, $f_A,f_B : \mathtt{CMem}\rightarrow\mathcal{S}(\mathcal{H})$ are the pre- and post-conditions.

The formula $\{f_A\} S \{f_B\}$ is true in the sense of partial correctness, written in $\models \{ f_A \} S \{ f_B \}$, if for any singleton cq-state $(m,\rho)$: 
$(m,\rho) \models f_A$ implies $\sem{S}(m,\rho) \models f_B$. Whenever $A,B\in AExp$, $\models\{A\}S\{B\}$ iff $\models\{\sem{A}\}S\{\sem{B}\}$.
\end{definition}

\begin{definition}[Weakest liberal precondition]\label{def-wlp}
For any program $S\in Prog$ and $f_B : \mathtt{CMem}\rightarrow\mathcal{S}(\mathcal{H})$, we define the function $wlp.S.f_B : \mathtt{CMem}\rightarrow\mathcal{S}(\mathcal{H})$ as:
$$wlp.S.f_B(m_{in}) \triangleq \bigwedge_{m_{out}} \ker\left(\sem{S}^{\prime\ast}_{m_{in},m_{out}}(\mathtt{P}_{f_B(m_{out})^\bot})\right)$$ 
where $\sem{S}^{\prime\ast}_{m_{in},m_{out}}$ is the dual super-operator of $\sem{S}_{m_{in},m_{out}}$, and $\ker$ is the kernal of linear operators as defined in Appendix \ref{review-subspace}.
$\models \{wlp.S.f_B\} S \{B\}$ and furthermore,
$wlp$ is well-defined in the sense that, for any $f_A$ such that $\models \{ f_A \} S \{ f_B \}$, it holds that $f_A\models wlp.S.f_B$.
\end{definition}

We first claim a technical lemma:
\begin{lemma}
For any density operator $\rho$, quantum operation $\cE$ and subspace $S$, we have:
    $$\Supp(\cE(\rho)) \subseteq S\mbox{ iff } \Supp(\rho) \subseteq \ker(\cE^\ast(\mathtt{P}_{S^\bot})).$$
\end{lemma}
\begin{proof}
Observe the following facts:
$$\Supp(A)\subseteq Q\mbox{ iff }\Tr(A\mathtt{P}_{Q^\bot}) = 0, \qquad
\Tr(AB) = 0\mbox{ iff }\Supp(A)\Supp(B) = 0$$
where $A, B$ are positive semi-definite operators, $Q$ is a subspace.
\begin{align*}
    &\Supp(\cE(\rho)) \subseteq S \Leftrightarrow\ 
    \Supp(\cE(\rho))\mathtt{P}_{S^\bot} = 0 \\ \Leftrightarrow\ 
    & \Tr(\cE(\rho)\mathtt{P}_{S^\bot}) = 0  \Leftrightarrow\ 
     \Tr(\rho\cE^\ast(\mathtt{P}_{S^\bot})) = 0 \\ \Leftrightarrow\ 
    & \Supp(\rho)\Supp(\cE^\ast(\mathtt{P}_{S^\bot})) = 0 \Leftrightarrow\ 
     \Tr(\rho\mathtt{P}_{\ker(\cE^\ast(\mathtt{P}_{S^\bot}))^\bot}) = 0 \\ \Leftrightarrow\ 
    & \Supp(\rho)\subseteq \ker(\cE^\ast(\mathtt{P}_{S^\bot}))
\end{align*}
\end{proof}
\begin{proof}[Proof of Definition \ref{def-wlp}]
    We show $\models \{wlp.S.f_B\} S \{B\}$ and the well-definedness as:
    \begin{align*}
        &\forall\, (m,\rho),\ (m,\rho)\models wlp.S.f_B \\ \Leftrightarrow\ 
        &\forall\, (m,\rho),\ \Supp(\rho)\subseteq wlp.S.f_B(m) \\ \Leftrightarrow\ 
        &\forall\, (m,\rho),\ \Supp(\rho)\subseteq \bigcap_{o} \ker\left(\sem{S}^{\prime\ast}_{m,o}(\mathtt{P}_{f_B(o)^\bot})\right) \\ \Leftrightarrow\ 
        &\forall\, (m,\rho), o,\ \Supp(\rho)\subseteq \ker\left(\sem{S}^{\prime\ast}_{m,o}(\mathtt{P}_{f_B(o)^\bot})\right) \\ \Leftrightarrow\ 
        &\forall\, (m,\rho), o,\ \Supp(\sem{S}^{\prime}_{m,o}(\rho))\subseteq f_B(o) \\ \Leftrightarrow\ 
        &\forall\, (m,\rho), \sum_o \sem{S}^{\prime}_{m,o}(\rho) \subseteq f_B(o) \\ \Leftrightarrow\ 
        &\forall\, (m,\rho), \sem{S}(m,\rho) \models f_B
    \end{align*}
    Since $(m,\rho)\models wlp.S.f_B$ must holds, so $f_A\models wlp.S.f_B$.
\end{proof}

As a corollary of the above proof, we have: 
\begin{corollary}
    For all $f_A, f_B$ and $S$, if for all $(m,\rho)$, $(m,\rho)\models f_A$ iff $\sem{S}(m,\rho)\models f_B$, then $f_A = wlp.S.f_B$.
\end{corollary}

To analyze the completeness of the proof system, it is necessary to explore the expressivity of the assertion language, that is, whether there exists an assertion semantically equivalent to the weakest precondition for the given postcondition which is expressed in the syntax.
\begin{theorem}[Weak definability]
\label{thm-definability}
    For any program $S\in Prog$ that does not contain while statements and post-condition $B\in AExp$, there exists an assertion $A\in AExp$, such that:
    $$\sem{A} = wlp.S.\sem{B}.$$
\end{theorem}
\begin{proof}
    We prove it by induction on the structure of the program $S$.
    \begin{itemize}
        \item $S\equiv \mathbf{skip}.$ By notice that $wlp.\mathbf{skip}.\sem{B} = \sem{B}$.
        \item $S\equiv \qut{U_1}{q_i}$ or $S\equiv \qut{U_2}{q_iq_j}$.
            Observe that $wlp.\qut{U_1}{q_i}.\sem{B} = U_{1i}^\dag \sem{B} U_{1i}$ and 
            $wlp.\qut{U_2}{q_iq_j}.\sem{B} = U_{2ij}^\dag \sem{B} U_{2ij}$. According to Theorem \ref{thm-pauli-closed}, in the case that $U_1\in\{X,Y,Z,H,S,T\}$ and $U_2\in\{CNOT,CZ,iSWAP\}$, $A$ is obtained by corresponding substitution of $p_r$ in B.
        \item $S\equiv x \coloneq e.$ By notice that $wlp.x \coloneq e.\sem{B} = \sem{B[e/x]}$.
        \item $S\equiv S_1\fatsemi S_2.$ By induction hypothesis, there exists $A_1$ such that $wlp.S_2.\sem{B} = \sem{A_1}$ and $A_2$ such that $wlp.S_1.\sem{A_2} = \sem{A_1}$. It is sufficient to show that $wlp.S_1.(wlp.S_2.f_B) = wlp.(S_1\fatsemi S_2).f_B$:
        
        \begin{align*}
            &wlp.(S_1\fatsemi S_2).f_B(i) \\
          =\ & \bigwedge_o \ker(\sum_m\sem{S_1}^{\prime\ast}_{i,m}(\sem{S_2}^{\prime\ast}_{m,o}(f_B(o)^\bot))) \\
          =\ & \bigwedge_o \Big(\bigvee_m\Supp(\sem{S_1}^{\prime\ast}_{i,m}(\sem{S_2}^{\prime\ast}_{m,o}(f_B(o)^\bot)))\Big)^\bot \\
          =\ & \bigwedge_m \ker(\sem{S_1}^{\prime\ast}_{i,m}(\big(\bigwedge_o\ker(\sem{S_2}^{\prime\ast}_{m,o}(f_B(o)^\bot))\big)^\bot)) \\
          =\ & \bigwedge_m \ker(\sem{S_1}^{\prime\ast}_{i,m}((wlp.S_2.f_B(m))^\bot)) \\
          =\ & wlp.S_1.(wlp.S_2.f_B)(i)
        \end{align*}
        We use the fact that $\Supp(\sum_i f_i) = \bigvee_i \Supp(f_i)$, $\Supp(\bigwedge S_i) = \bigwedge S_i$. We here for simplicity do not distinguish between subspace and its corresponding projection.
        \item $S \equiv \qassign{\mathbf{meas}[P]}{x}$. We show that:
        $$wlp.\qassign{\mathbf{meas}[P]}{x}.\sem{B} = \sem{(P \wedge B[0/x])\vee (\neg P \wedge B[1/x])}.$$
        For all $(m,\rho)$, we have:
        \begin{align*}
            & \qassign{\mathbf{meas}[P]}{x}(m,\rho) \models B \\ \Leftrightarrow\ 
            & (m[0/x], \mathtt{P}_{\sem{P}_m}\rho \mathtt{P}_{\sem{P}_m}) + 
            (m[1/x], \mathtt{P}_{\sem{P}_m^\bot}\rho \mathtt{P}_{\sem{P}_m^\bot}) \models B \\ \Leftrightarrow\ 
            & \sem{P}_m\sc \Supp(\rho)\subseteq \sem{B[0/x]}_m \mbox{ and }
            \sem{P}_m^\bot \sc \Supp(\rho)\subseteq \sem{B[1/x]}_m \\ \Leftrightarrow\ 
            & \Supp(\rho)\subseteq (\sem{P}_m \wedge \sem{B[0/x]}_m) \vee (\sem{P}_m^\bot \wedge \sem{B[1/x]}_m) \\ \Leftrightarrow\
            & (m,\rho) \models (P \wedge B[0/x])\vee (\neg P \wedge B[1/x]) 
        \end{align*}
        where the third and fourth lines are proved by employing properties of quantum logic. 
        \item $S\equiv \qqif{b}{S_1}{S_0}$. By induction hypothesis, there exists $A_0$ such that $wlp.S_0.\sem{B} = \sem{A_0}$ and $A_1$ such that $wlp.S_1.\sem{B} = \sem{A_1}$. It is sufficient to show that 
        $$wlp.\qqif{b}{S_1}{S_0}.f_B = \sem{(\neg b \wedge A_0) \vee (b\wedge A_1)}.$$
        For all $(m,\rho)$, by noticing that any singleton can only hold for one of the $\neg b \wedge A_0$ and $b\wedge A_1$, so we have:
        \begin{align*}
        & (m,\rho) \models (\neg b \wedge A_0) \vee (b\wedge A_1) \\ \Leftrightarrow\ 
        & (m,\rho) \models A_0 \mbox{ if $m(b) = \mathbf{false}$ or } (m,\rho) \models A_1 \mbox{ if $m(b) = \mathbf{true}$} \\ \Leftrightarrow\ 
        & \sem{\qqif{b}{S_1}{S_0}}(m,\rho) \models B  \mbox{ or } \sem{\qqif{b}{S_1}{S_0}}(m,\rho) \models B \\ \Leftrightarrow\ 
        & \sem{\qqif{b}{S_1}{S_0}}(m,\rho) \models B
        \end{align*}
        \item $S\equiv q_i \coloneq |0\>.$ Realize that initialization can be implemented by measurement and a controlled $X$ gate, i.e., 
        $$\sem{q_i \coloneq |0\>} = \sem{b := \mathbf{meas}[Z_i]\fatsemi \qqif{b}{q_i:=X}{\mathbf{skip}}},$$
        where assume that $b$ is some temporal variable and won't be considered in pre-/post-conditions. As such, we have: 
        $$wlp.q_i \coloneq |0\>.\sem{B} = \sem{(Z_i \wedge B) \vee (-Z_i \wedge B[-Y_i/Y_i, -Z_i/Z_i])}.$$
    \end{itemize}
\end{proof}

\subsection{Soundness and Weak Relative Completeness}
\label{sec:app-sound-complete}
We first claim the weak completeness of our proof system:
\begin{theorem}[Weak relative completeness]
\label{thm-weak-complete}
The proof system presented in Fig. 3 is relatively complete for finite QEC programs (without loops); that is, for any $A,B\in AExp$ and $S\in Prog$ that does not contain while statements, $\models \{A\} S \{ B\}$ implies $\vdash \{A\} S \{ B\}$.
\end{theorem}
With the help of Theorem \ref{thm-definability} and noticing that rules except for (While) and (Con) presented in Fig. 3 are in a backward way with exactly the weakest liberal preconditions, then Theorem \ref{thm-weak-complete} are a direct corollary. For Theorem 4.3, we only need to further prove the soundness of rules (While) and (Con), while, the latter is indeed trivial.

\begin{proof}[Proof of (While) for Theorem 4.3]
    By employing Proposition \ref{prop-struct}, it is sufficient to show that for any $(m,\rho)$ such that $(m,\rho)\models A$ and any $o\in\mathtt{CMem}$,
    \begin{align*}
        &\forall\, o\in\mathtt{CMem},\ \Supp(\lim_n \sem{(\mathbf{while})^n}(m,\rho)(o)) \subseteq \sem{\neg b\wedge A}_o \\ \Leftarrow\ 
        & \forall\, o\in\mathtt{CMem}, n, \Supp(\sem{(\mathbf{while})^n}(m,\rho)(o)) \subseteq \sem{\neg b\wedge A}_o \\ \Leftarrow\ 
        & \models \{A\} (\mathbf{while})^n \{\neg b\wedge A\}
    \end{align*}
    This can be proved by induction on $n$. For base case, $n = 0$, then $\sem{(\mathbf{while})^0}(m,\rho) = (m,0)$, so obviously satisfies $\sem{\neg b\wedge A}$. For induction step,
    \begin{align*}
        \models \{A\} (\qqif{b}{S\fatsemi(\mathbf{while})^n}{\mathbf{skip}} \{\neg b\wedge A\}
    \end{align*}
    by employing Theorem \ref{thm-definability}, we only need to show that:
    $$ A \models (b \wedge (b\wedge A)))\vee (\neg b \wedge (\neg b\wedge A)))$$
    which is trivial since $b, A$ commute with each other, and thus distribution law holds.
\end{proof}

\paragraph{Discussion on completeness}
Different from previous works that do not strictly introduce (countable) assertion language, the main obstacle is to show the expressivity of the assertion language. From a semantics view, it is straightforward to define the weakest liberal precondition $wlp.S.B$ for any program $S\in Prog$ with respect to postcondition $B\in AExp$ following from~\cite{zhou2019applied, feng2021quantum}. However, it remains to be proven that any $wlp.S.B$ is expressible in $AExp$, i.e., there exists $A\in AExp$ such that $\sem{A} = wlp.S.B$.
In classical and probabilistic program logic \cite{apt2010verification, BKK21}, the standard approach uses G\"odelization technique to encode programs and then prove the expressibility of the weakest precondition for loop statements. Unfortunately, due to the adoption of quantum logic, handling the while construct becomes much more challenging, and only a weak definability is proved above. 

\section{Explanation Omitted in Section 5.1}
\subsection{Explanation of Eqn. (8)}\label{eff-def-app}
The derivation of Eqn. (8) may require further explanation. We  consider the QEC program in the general case that is: 
\begin{equation}
\begin{aligned}
& \left\{\bigwedge_{i}g_i \wedge \bigwedge_j L_j \right\}\\
& \qqfor{1\cdots n}{\qut{X}{[x_i]q_i},\qut{Z}{[z_i]q_i}} \\
& \left\{\bigwedge_i(-1)^{c_i}g_i \wedge \bigwedge_{j}(-1)^{c_j} \bar{L}_j\right\} \\
\end{aligned}
\end{equation}
\begin{equation}
\begin{aligned}
& \qqfor{1\cdots n}{z_i = f_{z,i}(\bs),x_i = f_{x,i}(\bs)} \\
& \left\{\bigwedge_i(-1)^{r_i(\bs)}g_i \wedge \bigwedge_j(-1)^{r_j(\bs)}\bar{L}_j\right\} \\
& \qqfor{1\cdots n-k}{\qassign{M[g_i]}{s_i}}\fatsemi  \\
&  \left\{\bigvee_{\bs\in\{0,1\}^{n-k}}\bigwedge_i(-1)^{r_i(\bs)}g_i \wedge \bigwedge_j(-1)^{r_j(\bs)}\bar{L}_j\right\} \\
& \qqfor{1\cdots n}{\qut{U_1}{[e_i]q_i}}  \fatsemi \\
&\left\{\bigvee_{\bs\in\{0,1\}^{n-k}}\bigwedge_i(-1)^{r_i(\bs) + h_i(\be)}g_i' \wedge \bigwedge_j(-1)^{r_j(\bs) + h_j(\be)}\bar{L}_j'\right\} \\
\end{aligned}
\end{equation}

Here we obtain the desired form of verification condition; The functions $r_i(\bs), r_j(\bs)$ denotes the corrections made on operator $g_i, \bar{L}_j$ according to the syndromes $\bs$ and $h_i(\be)$ denotes the total (Pauli) errors injected to those operators. A complete program also needs to include the preparation of logic gates and (potentially) the errors propagated from the previous cycle. However, we notice that the unitary gates either change the Pauli operator or contribute to the error term in the phase. Therefore it is reasonable to conclude that generally, the verification should be in the form of Eqn. \eqref{wp-stabilizer}.

\paragraph{Explanation for case (2) in proof} The claim in (2) requires that $g_i', \bar{L}_j'$ do not depends on $\bs$ and $\be$. To see this, the first thing is correction operations and measurements will not change the stabilizers at all. Afterward, the implementation of logical operations does not contain conditional Pauli gates and, therefore does not introduce terms containing $s$ or $e$ in $g_i', \bar{L}_j'$. Finally, if any conditional non-Pauli errors are inserted before/after logical operations, then it will introduce terms involving $e$ in $g_i'$. However, changes of Paulis in $g_i'$, $L_j'$ caused by non-Pauli errors will induce non-commuting pairs with $g_i$, therefore violating the assumption that all $g_i, g_i', \bar{L}_j, \bar{L}_j'$ are commute to each other.

\subsection{Omitted Proof in Section 5.1}\label{eff-verify-app}

We give a formal proof for the proposition mentioned in Section 5.1.

\begin{proof}
\textit{Proof of I.} From \cite{sarkar2021sets} we know that for $n$-qubit Pauli expressions, the biggest commuting group has $2^n$ elements, which is generated by $n$ independent and commuting generators. We note this group generated by $\{P_1,\dots,P_n\}$ by $S$. Therefore, if $\exists i, P_i^{\prime} \neq \Pi_j P_{i_j} $ for any set of indices $\{i_j\}$ up to a phase, then $P_i'$ is not contained in $S$, which means that $P_i'$ anticommutes with some of the $P_j$. 

\textit{Proof of II.} We denote $S^{\prime} = \langle P_1^{\prime}, \dots, P_n^{\prime}\rangle$ and $V_S, V_{S^{\prime}}$ being the state space stabilized by $S, S^{\prime}$. It is easy to see that $V_S, V_S^{\prime}$ are of dimension 1~\cite[Chapter 10]{nielsen2010quantum}. Therefore since $\{P_1, \dots, P_n, P_1', \dots, P_n'\}$ are commute to each other, for $\qstate{\psi} \in V_S$, $P_i'\qstate{\psi} = \Pi_j P_{i_j}\qstate{\psi} = \qstate{\psi}$, which is $V_S = V_{S^\prime}$. Therefore:
\begin{equation}
\left((-1)^{b_1} P_1 \wedge \dots \wedge (-1)^{b_n} P_n \right) \wedge P_c \equiv (\bigwedge (-1)^{\sum_j b_{i_j}}\Pi_j P_{i_j}) \wedge P_c \equiv (\bigwedge_{i=1}^{n} (-1)^{\sum_j b_{i_j} + \alpha_i}P_i^{\prime}) \wedge P_c
\end{equation}
Moreover, for independent and commuting $\{P_1^{\prime}, \dots, P_n^{\prime} \}$, we have: 
\begin{equation}
 \left(\bigwedge_{i=1}^{n}b_i' = \sum_j b_{i_j} + \alpha_i\right) \wedge\left( (-1)^{\sum_{j}b_{1_j}}P_1^\prime \wedge \dots \wedge (-1)^{\sum_j b_{n_j}} P_n^\prime\right) \models 
\left( (-1)^{b_1'} P_1^\prime\wedge \dots \wedge (-1)^{b_n'} P_n^\prime\right) 
\label{vc-equiv-right}
\end{equation}
Therefore if $P_c \models \bigwedge_{i=1}^{n}(b_i' =\alpha_i + \sum_j b_{i_j})$, then 
\begin{equation}
P \equiv (\bigwedge_{i=1}^{n} (-1)^{\sum_j b_{i_j}}P_i^{\prime}) \wedge P_c \models \left( \bigwedge_{i=1}^{n}b_i' =\alpha_i + \sum_j b_{i_j}\right) \wedge (\bigwedge_{i=1}^{n} (-1)^{\sum_j b_{i_j}}P_i^{\prime}) \models P'
\end{equation}
Therefore we have finished the proof for \textit{II}. In fact we find that for independent and commuting generators $\{P_1^{\prime}, \dots, P_n^{\prime}\}$, the $\models$ is indeed $\equiv$ in Eqn. \eqref{vc-equiv-right}, therefore in our tool we directly transform the verification condition into the classical one in \textit{II}. 
\end{proof}
\section{Details in Case Study}\label{case-steane-app}
We have proposed the verification condition generated using inference rules in the main text, but we omit the derivation process. In this section we illustrate the derivation 
process of the verification condition mentioned in Section 5.2. 

\subsection{Details in Case I: Pauli Errors}\label{case1-steane-app}
We consider the case when implementing a logical Hadamard operation on a Steane code. The single Pauli error can propagate from the previous operation or occur after the logical 
gate. Therefore the program \textbf{Steane} is stated as in Table 1. 

Following this program we recall the correctness formula in Eqn. \ref{corr-steane}.
\begin{equation}
\begin{aligned}
& \left\{\big(\sum_{i=1}^{7} (e_i + e_{pi}) \leq 1\big) \wedge  \big((-1)^{b}\bar{X} \wedge (-1)^{0} g_1\wedge \cdots \wedge (-1)^0 g_6\big)\right\} \\
& \textbf{Steane}(Y,H) \quad \left\{(-1)^{b}\bar{Z} \wedge (-1)^{0} g_1\wedge \cdots \wedge (-1)^0 g_6\right\} \\
\end{aligned}
\label{corr-formula-steane}
\end{equation}
The correctness formula describes the condition that when there is at most 1 Pauli error (summing the errors occurring before and after the logical gate.) Then the correction can successfully output the correct state. 

According to~\cite{Fang2024symbolic}, to verify the correctness of the program we need to further consider the logical state after logical Hadamard gate as another postcondition. However we notice that the $X$ and $Z$ stabilizer generators and logical operators are the same, therefore only verifying the correctness for the postcondition in Eqn. \eqref{corr-formula-steane} is sufficient for Steane code. 

We prove Eqn. \eqref{corr-formula-steane} by deducing from the final postcondition to the forefront:
\begin{equation}
\begin{aligned}
& \left\{(-1)^{b}\bar{Z} \wedge (-1)^{0} g_1\wedge \cdots \wedge (-1)^0 g_6\right\}\\
& \qqfor{1\cdots 7}{\qut{X}{[x_i]q_i},\qut{Z}{[z_i]q_i}} \\
& \left\{(-1)^{b+c_0}\bar{Z} \wedge (-1)^{c_1} g_1\wedge \cdots \wedge (-1)^{c_6} g_6\right\} \\
& \qqfor{1\cdots 7}{\qassign{f_{z,i}(s_1,s_2,s_3)}{z_i},\qassign{f_{x,i}(s_4,s_5,s_6)}{x_i}} \\
& \left\{(-1)^{b+ r_7(\bs)}\bar{Z}\wedge (-1)^{r_1(\bs)}g_1 \wedge \cdots \wedge (-1)^{r_6(\bs)}g_6\right\} \\
& \qqfor{1\cdots 6}{\qassign{M[g_i]}{s_i}}\fatsemi  \\
&  \left\{\bigvee_{\bs\in\{0,1\}^6}(-1)^{b+ r_7(\bs)}\bar{Z}\wedge (-1)^{r_1(\bs)}g_1 \wedge \cdots \wedge (-1)^{r_6(\bs)}g_6\right\} \\
\end{aligned}
\end{equation}
\begin{equation}
\begin{aligned}
& \qqfor{1\cdots 7}{\qut{Y}{[e_i]q_i}}  \fatsemi \\
&\left\{\bigvee_{\bs\in\{0,1\}^6}(-1)^{b+r_7(\bs)+h_1^{\prime}(\be)}\bar{Z}\wedge (-1)^{r_1(\bs)+h_1(\be)}g_1 \wedge \cdots \wedge (-1)^{r_6(\bs)+ h_6(\be)}g_6\right\} \\
& \qqfor{1\cdots 7}{\qut{H}{q_i}}\fatsemi \\
& \left\{\bigvee_{\bs\in\{0,1\}^6}(-1)^{b+r_7(\bs)+h_7(\be)}\bar{X}\wedge (-1)^{r_1(\bs)+h_1(\be)}g_1' \wedge \cdots \wedge (-1)^{r_6(\bs)+ h_6(\be)}g_6'\right\} \\
\end{aligned}
\end{equation}
\begin{equation}
\begin{aligned}
& \qqfor{1\cdots 7}{\qut{Y}{[e_{p_i}]q_i}} \fatsemi \\
& \left\{\bigvee_{\bs\in\{0,1\}^6}(-1)^{b+r_7(\bs)+h_7(\be )+k_7\mathbf(ep)}\bar{X}\wedge (-1)^{r_1(\bs)+h_1(\be)+k_1(\mathbf{ep})}g_1' \wedge \cdots \wedge (-1)^{r_6(\bs)+ h_6(\be)+k_6(\mathbf{ep})}g_6'\right\} \\
\label{eqn:corr-steane-Y-final}
\end{aligned}
\end{equation}

We explain the symbols in the phases of Paulis in detail:
\begin{enumerate}
\item $b$ is the initial phase for logical operator $\bar{Z}$.
\item $c_i$ stands for the sum of correction indicators $\sum_{j}z_{j,i}$ or $\sum_{j}x_{j,i}$ leading to the flipping the corresponding Pauli expression $g_i$. For example, since $g_1 = X_1X_3X_5X_7$, then $c_1 = z_1 + z_3 + z_5 + z_7$.
\item $f_{z,i}$, $f_{x,i}$ assign the decoder outputs to correction indicators $z_i$ and $x_i$.
\item $r_i(\bs)$ denotes the sum of decoder outputs corresponding to $c_i$. For example, $r_1(\bs) = f_{z,1}(\bs) + f_{z,3}(\bs) + f_{z,5}(\bs) + f_{z,7}(\bs)$. Here we lift the variables of decoder functions to become all of $s_i$s, denoted by $\bs$.
\item $h_i(\be)$ denotes the sum of injected errors after logical Hadamard leading to the phase flip of the corresponding Pauli. Take $g_1$ and $g_4$ as examples, since $g_1 = X_1X_3X_5X_7$, $g_4 = Z_1Z_3Z_5Z_7$, and the error is $Y$ error which flips both $X$ and $Z$ stabilizers, $h_1(\be) = h_4(\be) = e_1 + e_3 + e_5 + e_7$.
\item $g_i^{\prime}$ denotes the stabilizer generators before the logical Hadamard gate. 
By direct computation of stabilizer generators, we find that $g_1^{\prime} = g_4, g_2^{\prime} = g_5, \cdots g_6^{\prime} = g_3$. On the other hand, the phases of $g_i^{\prime}$ can also be tracked. 
\item $k_i(\mathbf{ep})$ denotes the sum of errors propagated from previous operation, which also lead to the flip of the Pauli expression. For example, $k_i^{\prime}(\mathbf{be}) = \sum_{i=1}^{7}e_{p_i}, k_i(\mathbf{ep}) = e_{p_1} + e_{p_3} + e_{p_5} + e_{p_7}$. 
\end{enumerate}

The verification condition (VC) to be proved is derived from the precondition: 
\begin{equation}
\begin{aligned}
\label{VC-pauli}
& \left\{(\sum_{i=1}^{7} (e_i + e_{pi}) \leq 1) \wedge  ((-1)^{b}\bar{X} \wedge (-1)^{0} g_1\wedge \cdots \wedge (-1)^0 g_6)\right\} \\
&\left\{\bigvee_{\bs\in\{0,1\}^6}(-1)^{b+f_0(\bs)+E_0 + E_{p_0}}\bar{X}\wedge (-1)^{f_1(\bs)+E_1+E_{p_1}}g_1' \wedge \cdots \wedge (-1)^{f_6(\bs)+ E_6 + E_{p_6}}g_6'\right\}
\end{aligned}
\end{equation}
When confronted with this verification condition, generally we follow the verification framework proposed in Section 5.1 to deal with the generators $g_1,\cdots,g_6$, and $g_1^{\prime}, \cdots g_6^{\prime}$ here. 
For our Steane code example, from the computation in explanation (6) we find that since the stabilizer generators are symmetric, the correspondence of the generators can be easily found. Therefore the verification condition is equivalent with:
\begin{equation}
\bigg(\sum_{i=1}^{7} (e_i + e_{pi}) \leq 1\bigg)  \models \vee_{\bs\in\{0,1\}^6}\wedge_{i=0}^{6}\bigg(f_i(\bs) + E_i +E_{p_i} = 0\bigg)
\end{equation}
Assuming a minimum-weight decoder, we provide decoding conditions for the function call:
\begin{equation}
\left(\sum_{i=1}^{7} x_i \leq \sum_{i=1}^{7} (e_i+ e_{pi}) \right) \bigwedge \left(\sum_{i=1}^{7} z_i \leq \sum_{i=1}^{7} e_i+ \sum_{i=1}^{7}e_{pi})\right) \bigwedge \bigg(\wedge_{i=1}^{6}(f_i(\bs) = s_i)\bigg) 
\end{equation}
we can first obtain the value of $\bs = (s_1, \cdots, s_6)$ then use the decoding condition to obtain the exact value of $\{x_i\}$ and $\{z_i\}$. Take Z corrections as an example (X corrections here are symmetric, therefore we omit), the constraints for them are:
\begin{equation}
\label{eqn:decoder-res}
\left\{\begin{aligned}
& \sum_{i=1}^7 z_ i \leq 1 \\
& z_1+z_3+z_5+z_7 = s_1 \\
& z_2+z_3+z_6+z_7 = s_2 \\
& z_4+z_5+z_6+z_7 = s_3 \\
\end{aligned}
\right.
\end{equation}
In the case $(e_3 = 1)$ or $(e_{p_3} = 1)$, $s_1 = s_2 = 1, s_3 = 0$, therefore $z_3 = 1$ is the unique solution that satisfies Eqn. (\ref{eqn:decoder-res}). Finally, it is obvious that $f_0(\bs) + E_0 + E_{p_0} = \sum_{i=1}^7 (z_i+e_i +e_{p_i}) = 0$, so the correctness formula is successfully verified. However, any error patterns that violates the constraint $\left(\sum_{i=1}^{7} e_i + \sum_{i=1}^{7}e_{pi} \leq 1\right)$ would induce a logical error. For example the pattern $e_1 = 1, e_(p_2) = 1$ corresponds to the measurement syndrome $s_1 = s_2 = 1 = s_4 = s_5 = 1, s_3 = s_6 = 0$ too, but it will be identified by the decoder as $e_3$, thereby correcting the $3^{\rm{rd}}$ qubit and resulting in a logical error.

\subsection{Details in Case II: Non-Pauli Errors}\label{case2-steane-app}

In Section 5.1, we have proposed a heuristic algorithm which attempts to prove the correctness formula Eqn. \ref{wp-stabilizer} when there exists non-commuting pairs. 

We further provide an example to correct an $H$ error which is inserted after the logical operation.
\begin{example}[Correcting an H error on Steane code]
Suppose that $e_7 = 1$, then 
\begin{equation}
\label{wp-example}
\begin{aligned}
& (-1)^b\bar{Z}' = (-1)^b Z_1Z_2Z_3Z_4Z_5Z_6X_7, g_1' = X_1X_3X_5Z_7, g_2' = X_2X_3X_6Z_7, \\
& g_3' = X_4X_5X_6Z_7, g_4' = Z_1Z_3Z_5X_7, g_5 = Z_2Z_3Z_6X_7, g_6' = Z_4Z_5Z_6X_7
\end{aligned}
\end{equation}

In this case the weakest precondition obtained by the QEC program is 
\begin{equation}
\label{pre-cond2}
\left\{ \bigvee_{s_1,\cdots,s_6\in\{0,1\}}(-1)^{b+f(\bs)} \bar{Z}' \wedge (-1)^{s_1}g_1' \wedge \cdots \wedge  (-1)^{s_6}g_6' \right\}
\end{equation}
Where $f(\bs) = 0$ iff $(s_4,s_5,s_6) = (0,0,0)$, otherwise $f(\bs) = 1$. 
Compute the non-commuting set, we obtain $NC = C' = \{\bar{Z}', g_1',\cdots,g_6'\}$. Multiply the elements by $g_4'$, then $P^\prime$ becomes: 
\begin{equation}
\begin{aligned}
& P^\prime = \{\bigvee_{s_1,\cdots s_6\in\{0,1\}}(-1)^{b+f(\bs)+s_4}Z_2Z_4Z_6 \wedge (-1)^{s_1+s_4+1}Y_1Y_3Y_5Y_7 \wedge \\
& (-1)^{s_2+s_4+1}(Z_1Z_3X_4X_6)Y_5Y_7 \wedge (-1)^{s_3+s_4+1}(Z_1Z_5X_2X_6)Y_3Y_7\wedge \\
& (-1)^{s_4} Z_1Z_3Z_5X_7\wedge (-1)^{s_4+s_5}Z_1Z_2Z_5Z_6 \wedge (-1)^{s_4+s_6}Z_1Z_3Z_4Z_6\} \\
\end{aligned}
\end{equation}
Extract the items corresponding to $\bs = (1,1,1,0,0,0), (1,1,1,1,1,1)$ from the union in Eqn. (\ref{wp-example}), then these two terms form a subspace which eliminates the stabilizer $Z_1Z_3Z_5X_7$ since they differs only in the sign of $g_4'$. These two terms are: 
\begin{equation}
\begin{aligned}
& \{(-1)^bZ_2Z_4Z_6 \wedge Y_1Y_3Y_5Y_7 \wedge (Z_1Z_3X_4X_6)Y_5Y_7 \wedge \\
& (Z_1Z_5X_2X_6)Y_3Y_7 \wedge Z_1Z_2Z_5Z_6 \wedge Z_1Z_3Z_4Z_6 \wedge Z_1Z_3Z_5X_7\} \\
\end{aligned}
\end{equation}

\begin{equation}
\begin{aligned}
& \{(-1)^bZ_2Z_4Z_6 \wedge Y_1Y_3Y_5Y_7 \wedge (Z_1Z_3X_4X_6)Y_5Y_7 \wedge \\
& (Z_1Z_5X_2X_6)Y_3Y_7 \wedge Z_1Z_2Z_5Z_6 \wedge Z_1Z_3Z_4Z_6 \wedge -Z_1Z_3Z_5X_7\} \\
\end{aligned}
\end{equation}

Now the subspace is stabilized by $C' -\{g_4'\}$. We prove the stabilizer state in the precondition of Eqn. (\ref{corr-formula-steane}) is contained in this subspace. To this end, add $g_4$ to $C' - \{g_4'\}$ to form a complete stabilizer state $\hat{\rho}'$:
\begin{equation}
\begin{aligned}
& \hat{\rho}' = \{(-1)^bZ_2Z_4Z_6 \wedge Y_1Y_3Y_5Y_7 \wedge (Z_1Z_3X_4X_6)Y_5Y_7 \wedge \\
& (Z_1Z_5X_2X_6)Y_3Y_7 \wedge Z_1Z_2Z_5Z_6 \wedge Z_1Z_3Z_4Z_6 \wedge Z_1Z_3Z_5Z_7\} \\
\end{aligned}
\end{equation}
Again multiplying all elements by $g_4$ we obtain the generator set: 
\begin{equation}
\begin{aligned}
& \hat{\rho}' = \{(-1)^b Z_1Z_2Z_3Z_4Z_5Z_6Z_7\wedge  X_1X_3X_5X_7,\wedge X_4X_5X_6X_7 \wedge\\
& X_2X_3X_6X_7\wedge Z_2Z_3Z_6Z_7\wedge Z_4Z_5Z_6Z_7\wedge Z_1Z_3Z_5Z_7\} \\
\end{aligned}
\end{equation}
This corresponds to the stabilizer state in the precondition of Eqn. (\ref{corr-formula-steane}).

The good symmetry of Steane code ensures that only considering logical Z states is sufficient. In fact for arbitrary logical state stabilized by an additive Pauli predicate $a\bar{Z} + b\bar{X}$ ($|a|^2 + |b|^2 = 1$), the solution is to find $\hat{\rho}'_{X/Z}$ for logical X and Z respectively. The arbitrary logical state falls in the subspace formed by the superposition of these two stabilizer states. 
\end{example}
\section{Detailed Implementation of \veriq}\label{details-imp}
We provide details of \veriq, our tool for formal verification of QEC programs, which are ignored in the main text.

\subsection{Correctness Formula Generator}
Provided the theoretical results of the QEC code, e.g. the parity-check matrix and the code parameters (allow estimation for code distance), the correctness formula generator would first generate the program description for error correction, including error injection, syndrome measurement, external call of decoders and corrections. The stabilizer assertions and logical operators $\bar{X}_L$, $\bar{Z}_L$ will also be created. Afterwards we generate other parts of the program according to the implementations of fault-tolerant operations. We use a tuple $(x,z,n)$ to describe a single Pauli operator on $n$-th qubit, and the correspondence of $(x,z)$ and Paulis are $\{(0,0):I, (0,1):Z, (1,0):X, (1,1):Y\}$. We allow $x$ and $z$ to be classical expressions, therefore reserving space for future support of non-Pauli errors which lead to changes of not only phases but also Pauli constructs of stabilizers.  
\subsection{VC Generator}
The VC generator, as the core of the tool, is consisted of parser, interpreter and VC transformer. The parser is responsible for parsing the Hoare triple generated according to the QEC code and the requirements provided by the user. We implement the parser and the interpreter of AST in Python based on \texttt{Lark}~\cite{lark_parser}, a lightweight parser toolkit which uses LALR(1) to parse context-free grammars. We first establish the context-free grammar for correctness formula including the programs and assertions; Next we built customized interpreter using the \textit{Transformer} interface provided by \texttt{Lark}. For transversal unitary operations e.g. transversal logical gates or error injection and correction, we introduce 'for' sentence as a syntactic sugar for the sequential execution of those operations. We implemented the inference rules on the abstract syntax tree (AST) built upon the syntax of assertions and finally obtain the (expected) weakest precondition. We implement the VC transformer using the method mentioned in Section 5.1 to transform the hybrid classical-quantum assertion we obtain by the interpreter into a purely classical SMT formula containing classical program variables. 

\subsection{SMT Solver}
We introduce different SMT solvers for different aims. First, we use Z3 \cite{de2008z3} and its python interface as the encoder of the logical formula from the AST generated by the previous tool. Each variable including errors, corrections and syndromes are initially constructed as a \textit{BitVector} object with width 1. Automatic zero extension is performed whenever required, for example when dealing with the sum of errors and corrections when encoding the decoder's condition into the logical formula. Therefore we make integer addition and bit-wise addition compatible with each other. 

Afterwards, we will call other SMT solvers to parse the logical formula and check the satisfiability of it. For logical formula which includes quantifier forall $\forall$ (Exists $\exists$ quantifier will be naturally removed by the SMT solver), \texttt{CVC5} \cite{DBLP:conf/tacas/BarbosaBBKLMMMN22} is applied because it has the best efficiency for solving logical formula with quantifiers. In comparison to \texttt{Bitwuzla}, \texttt{CVC5} exhibits relatively weaker performance in validating bit-variable problems; thus, there exists a trade-off yet to be explored regarding which solver demonstrates superior efficacy. 

Our SMT checker supports parallelization, whose details will be discussed below. Specifically, the (symbolic) logic formula to be verified is initially generated on the bus and broadcast to the various parallel processes through global variables. Each process then substitutes the corresponding symbols in the formula with the enumerated values it receives, ultimately invoking the solver to resolve the modified formula.

\subsection{Parallelization}

In the verification task, we aim to verifying the capability of correction for any errors that satisfy the condition about number of errors and distance: 
\begin{equation}
\sum_{i = 1}^{n} e_i \leq \lfloor \frac{d-1}{2}\rfloor
\end{equation}

As demonstrated in the main text, for each error configuration, the time spent to check the satisfiability of corresponding SMT problem is 
double-exponential with respect to $d$, which turns out to be extremely time-consuming for SMT solvers to check the whole task at once. To address this, we designed a parallelization framework to split the verification task into multiple subtasks by dynamically enumerating selected free variables. To estimate the difficulty of each subtask, we design a heuristic function which serves as the termination condition for enumeration:
\begin{equation}
2d*N(\rm{ones}) + N(\rm{bits}) > n
\end{equation}
$N(\rm{ones})$ represents the occurrences of 1 and $N(\rm{bits})$ counts the number of enumerated bits. Enumeration stops if the heuristic function is satisfied, leaving the remaining portion to be solved by the SMT solver. For verification tasks of general properties, the parallel SMT solver will terminate the ongoing processes and cancel the tasks waiting to be checked if there is a counterexample, indicating that the implementation may exist errors. Then the counterexample would be produced to help find the potential errors in the implementation of codes or logical operations.

\end{document}